\documentclass[english,showpacs,preprint,superscriptaddress]{revtex4-1}
\usepackage{lmodern}
\usepackage[T1]{fontenc}
\usepackage{babel}
\usepackage{textcomp}
\usepackage{mathtools}
\usepackage{amsmath}
\usepackage{amsthm}
\usepackage{amssymb}
\usepackage{graphicx}
\usepackage[unicode=true,
 bookmarks=false,
 breaklinks=false,pdfborder={0 0 1},backref=false,colorlinks=false]
 {hyperref}

\makeatletter
\theoremstyle{plain}
\newtheorem{prop}{\protect\propositionname}

\usepackage{amsthm}\usepackage{latexsym}\usepackage{bm}\usepackage{amsfonts}

\allowdisplaybreaks
\usepackage{babel}
\usepackage{babel}

\usepackage{babel}

\makeatother

\providecommand{\propositionname}{Proposition}

\begin{document}
\title{Hall Transport of Charged Particles in Magnetic Disk Array}
\author{Linlin An}
\affiliation{School of Physics, Hefei University of Technology, Hefei, Anhui 230009,
China}
\author{Peifeng Fan}
\email{corresponding author: pffan@ahu.edu.cn}

\affiliation{School of Physics and Optoelectronic Engineering, Anhui University,
Hefei, Anhui 230601, China}
\begin{abstract}
The concept of periodic structures has driven the development of advanced
materials like photonic and phononic crystals. These metamaterials
typically rely on complex repeating units or meta-atoms, limiting
their adaptability after fabrication. To overcome this limitation,
we introduce the concept of metafields, which are repeating patterns
of local magnetic fields instead of material structures. Unlike metamaterials,
which consist of atoms arranged in structured patterns, metafields
focus on the patterns of fields alone, allowing for dynamic property
adjustments through external electric currents. This study explores
a specific metafield where the repeating pattern is the magnetic disk
(MD), defined by a uniform magnetic field perpendicular to its surface.
By arranging multiple MDs, we form a magnetic disk array (MDA) and
theoretically investigate the charged particle dynamics within this
array. Our analysis highlights Hall transport phenomena, such as Hall
diffusivity, conductivity, and thermal Hall effects. Using complex
variables, we derive the collision integral and Boltzmann equation
for particle distribution, applying perturbation methods and Fourier
analysis to calculate transport coefficients. Simulations reveal a
one-way preferential diffusion at the interface between MDAs with
opposing field directions, where diffusion intensity varies with the
particle movement direction. This highlights metafields' potential
for dynamic particle control applications.
\end{abstract}
\maketitle

\section{introduction}

The concept of crystals, which are defined by the periodic repetition
of atomic or molecular structures, has long been a cornerstone of
material science in condensed matter physics. This principle of repeating
structures has been extended into the design of modern materials,
giving rise to a range of applications. For instance, materials with
periodic variations in permeability and permittivity have enabled
the creation of photonic crystals \citep{Yablonovitch1987,John_1987,Joannopoulos_1997},
which exhibit properties such as negative refraction\citep{Veselago1968,Pendry1999,Shelby_2001,Smith_2000,Sakai2012}
and optical cloaking\citep{Pendry_2006,Cai_2007,Valentine_2009,Lee_2021},
sparking widespread interest in the manipulation of light. Similarly,
periodic elastic structures are the foundation of phononic crystals
\citep{Sigalas_1992}, which offer exciting possibilities for controlling
sound waves \citep{Liu-Science2000,Lu_2009,Pennec_2010}. These materials
fall into a broader class known as metamaterials---engineered materials
that possess extraordinary properties not found in natural substances
\citep{Walser2000,Walser2001}. Metamaterials are typically divided
into two categories: wave metamaterials and diffusion metamaterials
\citep{Zhang_NRP2023,Yang2024}. For example, photonic crystals are
classified as electromagnetic (wave) metamaterials, while phononic
crystals belong to the category of acoustic wave metamaterials. A
defining feature of traditional metamaterials is the use of basic
repeating elements, or \textquotedbl meta-atoms,\textquotedbl{} each
consisting of numerous molecules and atoms.

However, a significant drawback is that once fabricated, these metamaterials
have limited flexibility, as key parameters such as the lattice structure
and material properties are difficult to modify post-production. This
rigidity restricts their ability to adapt to changing external conditions
or demands \citep{Jia_2020,Torres_Huerta_2022}. To overcome this
limitation, we propose a new class of materials called metafields,
where the repeating elements are not material structures but local
magnetic fields. This approach offers a major advantage: the ability
to dynamically adjust the properties of these metafields through the
external control of electric currents. Unlike traditional metamaterials,
metafields can be fine-tuned in real time, providing an unprecedented
level of adaptability. In this study, we focus on a specific type
of meta-atom within a metafield, known as a magnetic disk (MD), which
is defined as a circular region with a uniform magnetic field perpendicular
to its surface. By arranging multiple MDs, we construct a magnetic
disk array (MDA). We theoretically investigate the evolution of charged
particles within the MDA and analyze the resulting transport properties.
In particular, we discuss various Hall transport phenomena, such as
Hall diffusivity, Hall conductivity, and thermal Hall effects, that
arise in this system. In the literature, the term \textquotedbl Hall\textquotedbl{}
is sometimes replaced by \textquotedbl odd\textquotedbl{} to reflect
the antisymmetric nature of certain transport tensors \citep{Fruchart_2023}.
For example, the Hall diffusion tensor exhibits opposing off-diagonal
components \citep{Hargus_2021}, which allows a longitudinal concentration
gradient to induce a transverse particle flux. Such Hall-related transport
behavior has attracted considerable interest \citep{Samajdar_2019,Chen2020,Poggioli_2023,VegaReyes2022},
especially in systems like active chiral fluids \citep{Lou_2022,Banerjee2022},
where phenomena such as Hall viscosity commonly occur.

In our approach, we treat the interaction between particles and MDA
as collisions. We derive the collision integral (or collision operator)
and the corresponding Boltzmann equation to describe the particle
distribution. To simplify this derivation, we employ a complex variable
formulation to represent the two-dimensional motion of particles,
which proves to be more efficient than traditional vector formalism.
This technique also allows for the straightforward derivation of single
and inverse collision processes. Using perturbation methods and Fourier
analysis, we compute the perturbation distribution of particles. Assuming
local equilibrium is characterized by a mono-kinetic distribution,
we calculate key transport coefficients, including the Hall diffusion
tensor, Hall conductivity tensor, and thermal Hall conductivity tensor.
Furthermore, we simulate the Hall diffusion equation to further understand
these transport properties. A particularly interesting finding emerges
when particles evolve in regions where two MDAs with opposing field
directions are present. At the interface between these two MDAs, we
identify a one-way preferential diffusion phenomenon, where diffusion
intensity varies depending on the direction of particle movement.
This highlights the unique transport properties introduced by metafields
and suggests potential for new applications in controlling particles.

The rest of the paper is organized as follows: In Sec.\,\ref{sec:the-2D-montion},
we introduce complex variables to describe the collision process of
a charged particle with a single MD, where we derive the scattering
angle and the inverse collision process. In Sec.\,\ref{sec:Boltzmann-eq-collision-integral},
we formulate the collision integral and the Boltzmann equation. The
linear perturbed distribution function is then derived in Sec.\,\ref{sec:linear solution of Boltzmann equation}.
Finally, in Sec.\,\ref{sec:typical-transport-coefficient}, we calculate
key transport coefficients, including the Hall diffusion tensor, Hall
conductivity tensor, and thermal Hall conductivity tensor.

\section{2D motion of charged particle in a background magnetic field\label{sec:the-2D-montion}}

\subsection{Motion of charged particles in a uniform magnetic field using complex
variables}

In this study, we focus on the motion of a charged particle in a 2D
$xoy$-plane under a background magnetic field, defined as $\boldsymbol{B}(\boldsymbol{r})=B_{z}(x,y)\boldsymbol{e}_{z}$.
Here, $\ensuremath{\boldsymbol{B}(\boldsymbol{r})}$ represents the
magnetic field, and $\boldsymbol{e}_{z}$ is the unit vector along
the $z$-direction. Given that vectors in a 2D plane can be described
by complex numbers, and the infinitesimal gyromotion of a charged
particle in a magnetic field can be easily represented by the multiplication
of complex numbers, we employ the complex representation method to
analyze the particle's motion in magnetic fields. Using this approach,
the particle's position $\boldsymbol{r}=x\boldsymbol{e}_{x}+y\boldsymbol{e}_{y}$
and velocity $\boldsymbol{v}=v_{x}\boldsymbol{e}_{x}+v_{y}\boldsymbol{e}_{y}$
are conveniently represented as $\ensuremath{r=x+iy}$ and $v=v_{x}+iv_{y}$,
respectively. In these expressions, $i$ denotes the unit imaginary
number, while $\ensuremath{\boldsymbol{e}_{x}}$ and $\boldsymbol{e}_{y}$
are the unit vectors along the $x$ and $y$ directions, respectively.
To distinguish the magnitudes $\boldsymbol{r}$ (or $r$) and $\boldsymbol{v}$
(or $v$), we employ the notations $\left|\boldsymbol{r}\right|$
(or $\left|r\right|$) and $\left|\boldsymbol{v}\right|$ (or $\left|v\right|$).
The same convention is applied to other physical quantities as well.

Using this method, the particle\textquoteright s equations of motion,
expressed in vector form as
\begin{equation}
\begin{cases}
\dot{\boldsymbol{r}}\left(t\right)=\boldsymbol{v}\left(t\right),\\
\dot{\boldsymbol{v}}\left(t\right)=\frac{qB_{z}}{mc}\boldsymbol{v}\times\boldsymbol{e}_{z},
\end{cases}\label{eq:Lorentz-force-vector}
\end{equation}
can be transformed into a complex form as
\begin{equation}
\begin{cases}
\dot{r}\left(t\right)=v\left(t\right)\\
\dot{v}\left(t\right)=i\omega_{c}v\left(t\right),
\end{cases}\label{eq:Lorentz-force-complex}
\end{equation}
where $c$ is the speed of light, and $q$, $m$ and $\omega_{c}=-qB_{z}/mc$
are real numbers representing the charge, mass and gyrofrequency of
the charged particle, respectively. Generally, $\omega_{c}$ is a
function of $r$, i.e., $\omega_{c}=\omega_{c}\left(r\right)$. However,
for a charged particle in a uniform magnetic field, $\omega_{c}$
is a constant. With the initial position $r_{0}$ and initial velocity
$v_{0}$ provided, equation (\ref{eq:Lorentz-force-complex}) can
be easily solved. The trajectory is then given by
\begin{equation}
r(t)=r_{0}+r_{c}\left(e^{i\omega_{c}t}-1\right),\label{eq:trajectory}
\end{equation}
where 
\begin{equation}
r_{c}=\frac{v_{0}}{i\omega_{c}}.\label{eq:rc}
\end{equation}
The detailed proofs of Eq.\,(\ref{eq:trajectory}) are shown in Appendix.\,\ref{sec:proofs}.
The the expression $r_{c}e^{i\omega_{c}t}$ in Eq.\,(\ref{eq:trajectory})
represents a rotation of $r_{c}$ by an angle $\omega_{c}t$, where
$\omega_{c}>0$ indicates a counterclockwise rotation and $\omega_{c}<0$
indicates a clockwise rotation. The term $r_{c}\left(e^{i\omega_{c}t}-1\right)$
represents the displacement over the time interval $\Delta t=t$.
Consequently, equation (\ref{eq:trajectory}) describes the gyromotion,
with $|r_{c}|$ being the corresponding gyroradius. The expression
$r_{0}-r_{c}=r(t)-r_{c}e^{i\omega_{c}t}$ denotes the position of
gyrocenter, which remains fixed during the gyromotion.

\subsection{Scattering of a charged particle in MD\label{subsec:Scattering-in-a-MD}}

We start from Eq.\,(\ref{eq:trajectory}) to determine how a charged
particle is scattered by a MD. Assuming $\theta=\omega_{c}t$, the
orbit of the charged particle can be written as
\begin{equation}
r(\theta)=r_{0}+r_{c}\left(e^{i\theta}-1\right).\label{eq:particle-orbit}
\end{equation}
Similarly, by analogy with Eq.\,(\ref{eq:particle-orbit}), the circle
boundary of the MD can be expressed as
\begin{equation}
r(\alpha)=r_{0}+r_{D}\left(e^{i\beta}-1\right),\label{eq:boundary-MD}
\end{equation}
where $\left|r_{D}\right|$ and $\left(r_{0}-r_{D}\right)$ are the
radius and center position of the MD. Here, $\beta$ is the parameter
of the circle boundary of MD. From Eqs.\,(\ref{eq:particle-orbit})
and (\ref{eq:boundary-MD}), we can observe that the initial positions
of the two circles are equal, i.e., $r(\theta)=r(\beta)=r_{0}$ when
$\theta=\beta=0$. Consequently, the incidence point where the particle
enters the magnetic field region is located at $r_{0}$. To determine
the scattering law, we only need to know the exit position or $e^{i\theta}$
(or $e^{i\beta}$) at the exit position. We will refer to $\theta$
as the scattering angle.

Combing Eq.\,(\ref{eq:particle-orbit}) and Eq.\,(\ref{eq:boundary-MD}),
we can obtain
\begin{equation}
\frac{e^{i\beta}-1}{e^{i\theta}-1}=z,\label{eq:basic formular}
\end{equation}
where $z$ is a complex number defined by
\begin{equation}
z\equiv\frac{r_{c}}{r_{D}}=\gamma e^{i\varphi},\label{eq:z=00003Dgamma-e^phi}
\end{equation}
with 
\begin{equation}
\gamma\equiv\frac{|r_{c}|}{|r_{D}|},\label{eq:gamma}
\end{equation}
and $\varphi$ being the angle through which $r_{D}$ rotates counterclockwise
to align with $r_{c}$. Using Eq.\,(\ref{eq:basic formular}), we
can derive

\begin{align}
 & e^{i\beta}=\frac{1-z}{1-z^{*}}=\frac{1-\gamma e^{i\varphi}}{1-\gamma e^{-i\varphi}},\label{eq:e^alpha}\\
 & e^{i\theta}=\frac{1-z^{-1}}{1-z^{*-1}}=\frac{\gamma-e^{-i\varphi}}{\gamma-e^{i\varphi}},\label{eq:e^theta}
\end{align}
after some algebraic manipulation. The detailed proofs of Eqs.\,(\ref{eq:e^alpha})-(\ref{eq:e^theta})
are provided in Appendix.\,\ref{sec:proofs}). Substituting Eq.\,(\ref{eq:z=00003Dgamma-e^phi})
into Eq.\,(\ref{eq:e^theta}), the term $e^{i\theta}$ can be rewritten
in another form as 
\begin{equation}
e^{i\theta}=\frac{\left(\gamma^{2}-2\gamma\mathrm{cos}\varphi+\mathrm{cos}2\varphi\right)+i\left(2\gamma\mathrm{sin}\varphi-\mathrm{sin}2\varphi\right)}{\gamma^{2}-2\gamma\mathrm{cos}\varphi+1}.\label{eq:e^theta-phi-split}
\end{equation}
Using Euler's formula, $e^{i\theta}=\mathrm{cos}\theta+i\mathrm{sin}\theta$,
equation (\ref{eq:e^theta-phi-split}) can be equivalently split into
two equations
\begin{align}
 & \mathrm{cos}\theta=\frac{\gamma^{2}-2\gamma\mathrm{cos}\varphi+\mathrm{cos}2\varphi}{\gamma^{2}-2\gamma\mathrm{cos}\varphi+1},\label{eq:cos(theta)-phi}\\
 & \mathrm{sin}\theta=\frac{2\gamma\mathrm{sin}\varphi-\mathrm{sin}2\varphi}{\gamma^{2}-2\gamma\mathrm{cos}\varphi+1}.\label{eq:sin(theta)-phi}
\end{align}
Equations (\ref{eq:e^theta-phi-split})-(\ref{eq:sin(theta)-phi})
indicate that the scattering angle $\theta$ is only determined by
$z\equiv r_{c}/r_{D}$ or equivalently by the ratio $\gamma=|r_{c}|/|r_{D}|$
and the angle $\varphi$.

Suppose the initial velocity $v_{0}$ is along the positive direction
of the $x$-axis, i.e., $v_{0}=\left|v_{0}\right|$. Let 
\begin{equation}
r_{D}\equiv a+ib\label{eq:rM}
\end{equation}
 and using Eq.\,(\ref{eq:z=00003Dgamma-e^phi}), we obtain
\begin{align}
e^{i\varphi} & =\gamma^{-1}\frac{r_{c}}{r_{D}}=\frac{r_{c}/\left|r_{c}\right|}{r_{D}/\left|r_{D}\right|}=\frac{v_{0}/i\omega_{c}}{\left|v_{0}/i\omega_{c}\right|}\frac{r_{D}^{*}}{\left|r_{D}\right|}\nonumber \\
 & =\frac{\left|\omega_{c}\right|}{\omega_{c}}\frac{-ir_{D}^{*}}{\left|r_{D}\right|}=\mathrm{sgn}\left(\omega_{c}\right)\frac{-b-ia}{\left|r_{D}\right|},\label{eq:e^phi}
\end{align}
where we used the definition of $r_{D}$ (see Eq.\,(\ref{eq:rc})).
Here $\mathrm{sgn\left(\omega_{c}\right)}$ is the sign function about
$\omega_{c}$ which is defined by 
\begin{equation}
\mathrm{sgn}\left(\omega_{c}\right)=\begin{cases}
+1 & \omega_{c}>0,\\
-1 & \omega_{c}<0.
\end{cases}\label{eq:sgn(wc)}
\end{equation}
Applying Euler\textquoteright s formula again, equation (\ref{eq:e^phi})
can be rewritten as 
\begin{align}
 & \mathrm{cos}\varphi=-\mathrm{sgn}\left(\omega_{c}\right)\frac{b}{\left|r_{D}\right|},\label{eq:cos(phi)}\\
 & \mathrm{sin}\varphi=-\mathrm{sgn}\left(\omega_{c}\right)\frac{a}{\left|r_{D}\right|}.\label{eq:sin(phi)}
\end{align}

\subsection{The inverse collision in a MD \label{subsec:inverse-collision}}

In this study, we consider a charged particle crossing the region
of a MD as a 2D binary collision process. In the collision process,
the MD remains stationary, and the particle's momentum is not conserved;
however, its energy is conserved. To derive the collision integral
of the Boltzmann equation (see Sec.\,\ref{sec:Boltzmann-eq-collision-integral}),
we need to determine the inverse of the collision process. Specifically,
we need to know the velocity at which the particle must enter the
MD so that it can escape the region with a given initial velocity.
To clearly convey the fundamental physical process, we represent the
collision as
\begin{equation}
\left(\text{intial state};\mathrm{parameters}\right)\rightarrow\left(\text{final state};\mathrm{parameters}\right),
\end{equation}
where the state refers to the position and velocity of a particle,
i.e., a state is a 2-tuple $\left(r,v\right)$, and the parameters
are the invariants of the scattered particle and target object, such
as mass, charge, field strength, etc. We use the right arrow $"\rightarrow"$
to denote a possible process, and $"\nrightarrow"$ to indicate a
forbidden process.

We now define several basic discrete transformations of the collision
process. The time reversal transformation is defined as
\begin{equation}
T\left[\left(r_{A},v_{A}\right)\rightarrow\left(r_{B},v_{B}\right)\right]=\left(r_{B},-v_{B}\right)\rightarrow\left(r_{A},-v_{A}\right),\label{eq:time-reversal}
\end{equation}
where we omit the parameters for simplicity. From Eq.\,(\ref{eq:time-reversal}),
we can see that time reversal implies a reversed collision process;
that is, if we reverse the direction of the final velocity, the particle
will retrace its path and return to its original position. The parity
transformation is defined as
\begin{equation}
P\left[\left(r_{A},v_{A}\right)\rightarrow\left(r_{B},v_{B}\right)\right]=\left(-r_{A},-v_{A}\right)\rightarrow\left(-r_{B},-v_{B}\right),\label{eq:parity}
\end{equation}
which represents a centrosymmetric collision process. The mirror transformation
of the collision process about the $x_{i}$-axis, where $x_{i}=x,y$,
is collision process in a mirror perpendicular to $x_{i}$-axis. This
transformation is denoted by $P_{x_{i}}$ and is defined as follows:
\begin{align}
 & P_{x}\left[\left(r_{A},v_{A}\right)\rightarrow\left(r_{B},v_{B}\right)\right]=\left(-r_{A}^{*},-v_{A}^{*}\right)\rightarrow\left(-r_{B}^{*},-v_{B}^{*}\right),\label{eq:mirror-x}\\
 & P_{y}\left[\left(r_{A},v_{A}\right)\rightarrow\left(r_{B},v_{B}\right)\right]=\left(r_{A}^{*},v_{A}^{*}\right)\rightarrow\left(r_{B}^{*},v_{B}^{*}\right),\label{eq:mirror-y}
\end{align}
where the asterisk $*$ as a superscript denotes the conjugation of
the associated complex number. The rotation of the collision process
around the $x_{i}$ by an angle $\theta$, denoted as $R_{x_{i}}^{\theta}$,
is defined by 
\begin{equation}
R_{x_{i}}^{\theta}\left[\left(r_{A},v_{A}\right)\rightarrow\left(r_{B},v_{B}\right)\right]=\left(R_{x_{i}}^{\theta}\left(r_{A}\right),R_{x_{i}}^{\theta}\left(v_{A}\right)\right)\rightarrow\left(R_{x_{i}}^{\theta}\left(r_{B}\right),R_{x_{i}}^{\theta}\left(v_{B}\right)\right),\label{eq:rotation}
\end{equation}
where $R_{x_{i}}^{\theta}\left(r_{P}\right)$ and $R_{x_{i}}^{\theta}\left(v_{P}\right)$
(with $P=A,B$ ) represent the rotation of $r_{P}$ and $v_{P}$ by
an angle $\theta$. The sign-changing operators $\sigma_{f}$ and
$\sigma_{c}$ are defined as
\begin{align}
 & \sigma_{f}\left(F\right)=-F,\label{eq:sign-changing-F}\\
 & \sigma_{c}\left(q\right)=-q,\label{eq:sign-changing-q}
\end{align}
where $F$ represents the background field, such as an electric or
magnetic field. In Eqs.\,(\ref{eq:time-reversal})-(\ref{eq:rotation}),
we have assumed that all the collision processes are allowed.

Before discussing the inverse collision of a charged particle in a
MD, let's first revisit the scenario of Coulomb collisions. Suppose
the target particle is fixed during the collision process, meaning
the mass of the target particle is infinite. Figure.\,\ref{Fig: Coulomb and MDA collision}(a)
illustrates the Coulomb collision process. For any collision process
$\left(r_{A},v_{A}\right)\rightarrow\left(r_{B},v_{B}\right)$, ,
it is straightforward to verify that both time reversal symmetry (\ref{eq:time-reversal})
symmetry and parity symmetry (\ref{eq:parity}) are preserved in the
Coulomb collision process. To solve the inverse collision problem,
we can apply both time reversal $(T)$ and $(P)$ transformations.
This yields:
\begin{align}
 & TP\left[\left(r_{A},v_{A}\right)\rightarrow\left(r_{B},v_{B}\right)\right]=T\left[\left(-r_{A},-v_{A}\right)\rightarrow\left(-r_{B},-v_{B}\right)\right]\nonumber \\
 & =\left(-r_{B},v_{B}\right)\rightarrow\left(-r_{A},v_{A}\right).\label{eq:TP-Columb}
\end{align}
From Eq.\,(\ref{eq:TP-Columb}), we can see that if we place the
charge at $-r_{B}$ with velocity $v_{B}$, the final velocity will
be $\boldsymbol{v}_{A}$, which is exactly the inverse of the Coulomb
collision, as shown in Fig.\,\ref{Fig: Coulomb and MDA collision}(a).
It is also easy to verify that exchanging the order of operations
for $T$ and $P$, i.e., $TP=PT$, holds true for the Coulomb collision
process.

\begin{figure}
\includegraphics[scale=0.5]{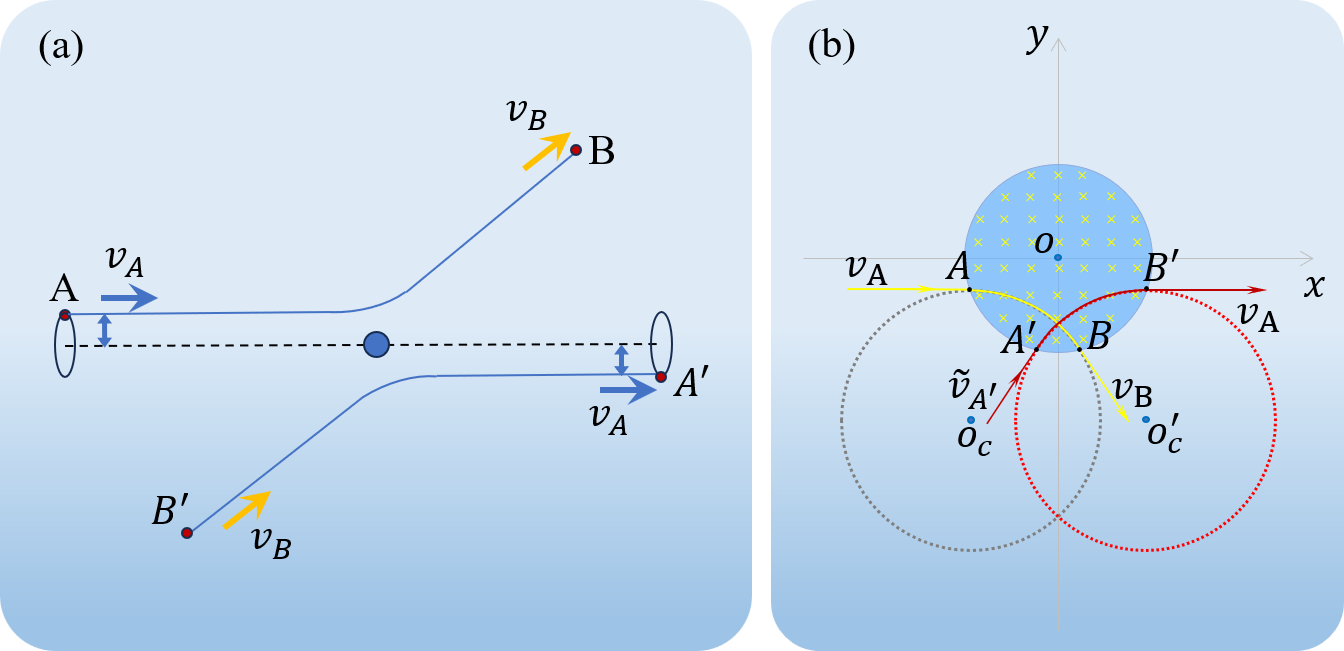}

\caption{\label{Fig: Coulomb and MDA collision} (a) Coulomb collision process.
The charged particles are scattered by a heavy, charged target, which
remains fixed during the collision. The process $B'\rightarrow A'$
is the inverse of the process $A\rightarrow B$. In this context,
the positions of $A$ and $B$ are $r_{A}$ and $r_{B}$, respectively,
while the positions of $A'$ and $B'$ are $-r_{A}$ and $-r_{B}$.
(b) Inverse collision process in a single MD. The circular area with
the cross symbolizes a magnetic disk. The red circular orbit represents
the mirror image of the yellow one with respect to the $y$-axis.
When a particle enters the magnetic field with velocity $v_{A}$,
parallel to the $x$-axis, it exits with velocity $v_{B}$. When the
incident velocity is $\tilde{v}$, which is the mirror of $v_{B}$
about the $y$-axis, the exit velocity is $v_{A}$.}
\end{figure}

For a particle crossing the MD, time reversal symmetry is broken due
to the background magnetic field. Specifically,

\begin{equation}
T\left[\left(r_{A},v_{A}\right)\rightarrow\left(r_{B},v_{B}\right)\right]=\left(r_{B},-v_{B}\right)\nrightarrow\left(r_{A},-v_{A}\right).
\end{equation}
This means that if we reverse the final velocity to $-\boldsymbol{v}_{B}$,
the particle will not return to its initial state. To reverse the
particle's trajectory, we need to simultaneously change the sign of
the magnetic field (or apply the sign-changing operator $\sigma_{f}$)
while reversing the final velocity. This can be expressed as:
\begin{equation}
\left(T\sigma_{f}\right)\left[\left(r_{A},v_{A};q,B_{z}\right)\rightarrow\left(r_{B},v_{B};q,B_{z}\right)\right]=\left(r_{B},-v_{B};q,-B_{z}\right)\rightarrow\left(r_{A},-v_{A};q,-B_{z}\right).\label{eq:(Tsigma_f)-MD}
\end{equation}
Additionally, since rotational symmetry always holds, we can rotate
the entire system around the axis-$y$ by an angle of $\pi$ to obtain
another possible collision process
\begin{align}
 & R_{y}^{\pi}\left[\left(r_{A},v_{A};q,B_{z}\right)\rightarrow\left(r_{B},v_{B};q,B_{z}\right)\right]\nonumber \\
 & =\left(R_{y}^{\pi}\left(r_{A}\right),R_{y}^{\pi}\left(v_{A}\right);q,-B_{z}\right)\rightarrow\left(R_{y}^{\pi}\left(r_{B}\right),R_{y}^{\pi}\left(v_{B}\right);q,-B_{z}\right).\label{eq:Ry(pi)-MD}
\end{align}
Suppose the incident velocity is parallel to the $x$-axis, it is
straightforward to observe that $R_{y}^{\pi}\left(r_{P}\right)=P_{x}\left(r_{P}\right)$
and $R_{y}^{\pi}\left(v_{P}\right)=P_{x}\left(v_{P}\right)$. This
implies that $R_{y}^{\pi}=P_{x}\sigma_{f}$. Therefore, we have
\begin{equation}
\left(P_{x}\sigma_{f}\right)\left[\left(r_{A},v_{A};q,B_{z}\right)\rightarrow\left(r_{B},v_{B};q,B_{z}\right)\right]=\left(-r_{A}^{*},-v_{A}^{*};q,-B_{z}\right)\rightarrow\left(-r_{B}^{*},-v_{B}^{*};q,-B_{z}\right).\label{eq:Px-sigma_f}
\end{equation}
We can now combine the $T\sigma_{f}$ and $P_{x}\sigma_{f}$ operations
to obtain the inverse velocity 
\begin{align}
 & \left(T\sigma_{f}\right)\left(P_{x}\sigma_{f}\right)\left[\left(r_{A},v_{A};q,B_{z}\right)\rightarrow\left(r_{B},v_{B};q,B_{z}\right)\right]\nonumber \\
 & =\left(T\sigma_{f}\right)\left[\left(-r_{A}^{*},-v_{A}^{*};q,-B_{z}\right)\rightarrow\left(-r_{B}^{*},-v_{B}^{*};q,-B_{z}\right)\right]\nonumber \\
 & =\left(-r_{B}^{*},v_{B}^{*};q,B_{z}\right)\rightarrow\left(-r_{A}^{*},v_{A}^{*};q,-B_{z}\right)=\left(-r_{B}^{*},v_{B}^{*};q,B_{z}\right)\rightarrow\left(-r_{A}^{*},v_{A};q,-B_{z}\right),\label{eq:Tsigma_f+Px-sigma_f}
\end{align}
where we used $v_{A}=v_{A}^{*}$ in the last step because $v_{A}$
is parallel to the $x$-axis. This inverse collision can also be achieved
by combining the operations $T\sigma_{c}$ and $P_{x}\sigma_{c}$.
This can be expressed as follows:
\begin{align}
 & \left(T\sigma_{c}\right)\left(P_{x}\sigma_{c}\right)\left[\left(r_{A},v_{A};q,B_{z}\right)\rightarrow\left(r_{B},v_{B};q,B_{z}\right)\right]\nonumber \\
 & =\left(T\sigma_{c}\right)\left[\left(-r_{A}^{*},-v_{A}^{*};-q,B_{z}\right)\rightarrow\left(-r_{B}^{*},-v_{B}^{*};-q,B_{z}\right)\right]\nonumber \\
 & =\left(-r_{B}^{*},v_{B}^{*};q,B_{z}\right)\rightarrow\left(-r_{A}^{*},v_{A}^{*};q,B_{z}\right)=\left(-r_{B}^{*},v_{B}^{*};q,B_{z}\right)\rightarrow\left(-r_{A}^{*},v_{A};q,-B_{z}\right).\label{eq:35}
\end{align}
The corresponding process is illustrated in Fig.\,\ref{Fig: Coulomb and MDA collision}(b).
It's also easy to verify that $\left(T\sigma_{f}\right)\left(P_{x}\sigma_{f}\right)=\left(P_{x}\sigma_{f}\right)\left(T\sigma_{f}\right)$
and $\left(T\sigma_{c}\right)\left(P_{x}\sigma_{c}\right)=\left(P_{x}\sigma_{c}\right)\left(T\sigma_{c}\right)$.
From Eq.\,(\ref{eq:Tsigma_f+Px-sigma_f}) or (\ref{eq:35}), we can
place the charged particle at $-r_{B}^{*}$ with the velocity $v_{B}^{*}$
ensuring that the particle's velocity will be $\boldsymbol{v}_{A}$
when it exits the MD. The inverse velocity, denoted as $\tilde{v}_{A'}$,
is then given by $\tilde{v}_{A'}=v_{B}^{*}$. From Sec.\,\ref{subsec:Scattering-in-a-MD},
we know that $v_{B}=e^{i\theta}v_{A}$ (see Eq.\,(\ref{eq:e^theta})
for the expression of $e^{i\theta}$ ). Therefore, we have
\begin{equation}
\tilde{v}_{A'}=v_{B}^{*}=e^{-i\theta}v_{A}.\label{eq:inverse-velocity}
\end{equation}

\section{collision integral and Boltzmann equation \label{sec:Boltzmann-eq-collision-integral}}

\subsection{Collision integral}

We now study the kinetics of a collection of particles scattered by
the MDA. Assuming all the MDs are identical, having the same radius
and magnetic field strength, the number density of the MDs is denoted
by $n_{D}(t,r)$. We proceed to derive the collision operator for
the particles scattered by the magnetic disk array. Consider charged
particles being scattered from velocity $v$ to another velocity,
referred to as \textquotedbl losses,\textquotedbl{} at position $r$
over a time interval $dt$. The number of collisions is then given
by $\left[d^{2}rd^{2}vf_{\alpha}\left(t,r,v\right)\right]\left[n_{D}(t,r)|v|dtdb\right]$,
where $f_{\alpha}\left(t,r,v\right)$ is the distribution function
for $\alpha$-species, and $b$, as defined by Eq.\,(\ref{eq:rM}),
is the impact parameter (see Fig.\,\ref{Fig: Coulomb and MDA collision}(a)).
The total number of collisions between particles and the MDs within
the volume element $d^{2}rd^{2}v$ during the time interval $dt$
is
\begin{equation}
dtd^{2}rd^{2}v\int_{-|r_{D}|}^{|r_{D}|}f_{\alpha}\left(t,r,v\right)n_{D}(t,r)|v|db.\label{eq:losses}
\end{equation}
The ``gains'', or the collisions that bring particles into the the
range $v$ during $dt$, are given by
\begin{align}
 & dtd^{2}rd^{2}\tilde{v}\int_{-|r_{D}|}^{|r_{D}|}f_{\alpha}\left(t,r,\tilde{v}\right)n_{D}(t,r)|\tilde{v}|db\nonumber \\
 & =dtdrdv\int_{-|r_{D}|}^{|r_{D}|}f_{\alpha}\left(t,r,\tilde{v}\right)n_{D}(t,r)|v|db,\label{eq:gains}
\end{align}
where $\tilde{v}$ is the inverse velocity of $v$. In deriving Eqs.\,(\ref{eq:losses})
and (\ref{eq:gains}), we used the relations $d\tilde{v}=dv$ and
$|\tilde{v}|=|v|$ (since the Lorentz force in a magnetic field does
not change the speed of the charged particle). Additionally, we assumed
that the particle's position is approximately the same as the center
of the MD when they encounter each other. Subtracting the ``losses''
(\ref{eq:losses}) from the ``gains'' (\ref{eq:gains}) and divide
it by $dtdrdv$, we thus derived the collision integral as

\begin{equation}
\left(\frac{\partial f_{\alpha}}{\partial t}\right)_{c}=n_{D}(t,r)|v|\int_{-|r_{D}|}^{|r_{D}|}\left[f_{\alpha}\left(t,r,\tilde{v}\right)-f_{\alpha}\left(t,r,v\right)\right]db.\label{eq:collision-integral-db}
\end{equation}

Since the inverse velocity $\tilde{v}$ is determined by $\varphi$
(see Eqs.\,(\ref{eq:inverse-velocity}) and (\ref{eq:e^theta-phi-split})),
we need to transform the integral in Eq.\,(\ref{eq:collision-integral-db})
into a form that is expressed in terms of $\varphi$. From Eq.\,(\ref{eq:cos(phi)}),
we can see that the scattering property is totally different for different
sign of gyro frequency. Using Eq.\,(\ref{eq:cos(phi)}), we have
\begin{equation}
db=s_{\alpha}|r_{D}|\mathrm{sin\varphi d\varphi},\label{eq:db}
\end{equation}
where 
\begin{equation}
s_{\alpha}\equiv\mathrm{sgn(\omega_{c\alpha})}.\label{eq:s_alpha}
\end{equation}
In this section and the ones that follow, we will use the subscript
$\alpha$ to distinguish between different physical quantities. For
example, $\omega_{c\alpha}$ in Eq.\,(\ref{eq:s_alpha}) represent
the gyrofrequencies of different particle species, given by $\omega_{c\alpha}=-q_{\alpha}B_{z}/m_{\alpha}c$,
where $q_{\alpha}$ and $m_{\alpha}$ denote the charge and mass of
the $\alpha$-species.
\begin{figure}
\includegraphics[scale=0.6]{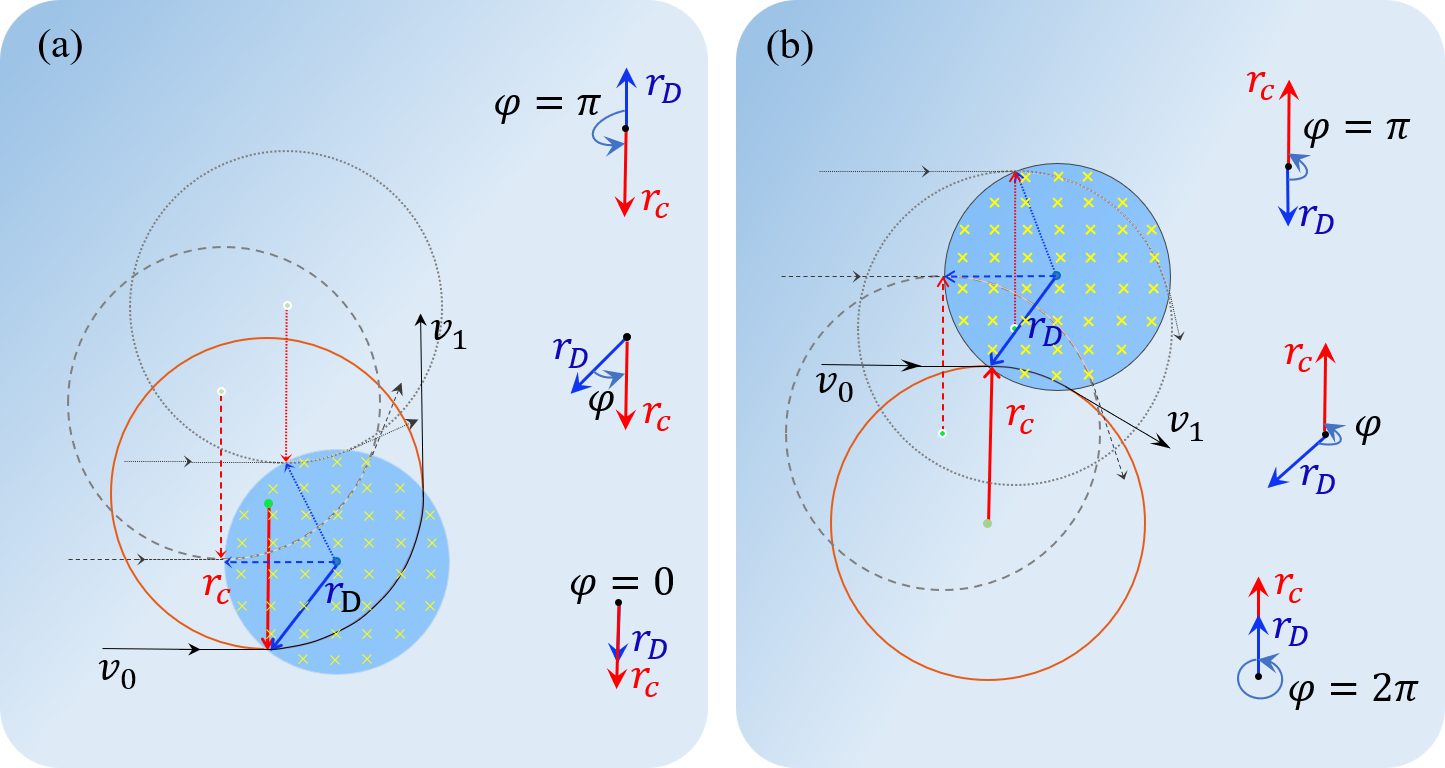}

\caption{\label{fig:Value-range-of-phi}Value range of $\varphi$ for two types
of collisions in a MD. (a) When $\omega_{c\alpha}>0$, charged particles
are deflected by magnetic fields, with $\varphi$ varying from 0 to
$\pi$; (b) When $\omega_{c\alpha}<0$, charged particles are deflected
by magnetic fields, with $\varphi$ varying from $\pi$ to $2\pi$.}
\end{figure}

In the case where $\omega_{c\alpha}>0$ (or $s_{\alpha}>0$, i.e.,
$q_{\alpha}>0$ and $B_{z}>0$, or $q_{\alpha}<0$ and $B_{z}<0$),
we have $\varphi=0$ when $b=-\left|r_{D}\right|$, and $\varphi=\pi$
when $b=\left|r_{D}\right|$ as shown in Fig.\,\ref{fig:Value-range-of-phi}(a).
By combining Eqs.\,(\ref{eq:e^theta}) and (\ref{eq:inverse-velocity}),
and substituting Eq.\,(\ref{eq:db}) into (\ref{eq:collision-integral-db}),
the collision integral becomes
\begin{equation}
\text{\ensuremath{\left(\frac{\partial f_{\alpha}}{\partial t}\right)}}_{c}^{\left(+\right)}=n_{D}(t,r)|v||r_{D}|\int_{0}^{\pi}\left[f_{\alpha}\left(t,r,\frac{\gamma_{\alpha}-e^{i\varphi}}{\gamma_{\alpha}-e^{-i\varphi}}v\right)-f_{\alpha}\left(t,r,v\right)\right]\mathrm{sin\varphi}d\varphi,\label{eq:collision-integral-0-pi(wc>0)}
\end{equation}
where $\gamma_{\alpha}$ is defined by Eq.\,(\ref{eq:gamma}) for
the $\alpha$-species, and the superscript $\left(+\right)$ is used
to indicate that $\omega_{c\alpha}>0$. Figure.\,\ref{fig:Value-range-of-phi}(b)
shows the other case where $\omega_{c\alpha}<0$ (or $s_{\alpha}<0$,
i.e., $q_{\alpha}<0$ and $B_{z}>0$, or $q_{\alpha}>0$ and $B_{z}<0$).
In this case, $\varphi=\pi$ when $b=-\left|r_{B}\right|$, and $\varphi=2\pi$
when $b=\left|r_{D}\right|$. Under this condition, the collision
integral can be expressed as
\begin{equation}
\text{\ensuremath{\left(\frac{\partial f_{\alpha}}{\partial t}\right)}}_{c}^{\left(-\right)}=n_{D}(t,r)|v||r_{D}|\int_{2\pi}^{\pi}\left[f_{\alpha}\left(t,r,\frac{\gamma_{\alpha}-e^{i\varphi}}{\gamma_{\alpha}-e^{-i\varphi}}v\right)-f_{\alpha}\left(t,r,v\right)\right]\mathrm{sin\varphi}d\varphi\label{eq:collision-integral-2pi-pi (wc<0)}
\end{equation}
by substituting Eq.\,(\ref{eq:db}) into (\ref{eq:collision-integral-db}),
where the superscript $\left(-\right)$ is used to indicate that $\omega_{c\alpha}<0$.
By changing the variable $\varphi\rightarrow\varphi-\pi$, equation
(\ref{eq:collision-integral-2pi-pi (wc<0)}) can be transformed into
\begin{equation}
\text{\ensuremath{\left(\frac{\partial f_{\alpha}}{\partial t}\right)}}_{c}^{\left(-\right)}=n_{D}(t,r)|v||r_{D}|\int_{0}^{\pi}\left[f_{\alpha}\left(t,r,\frac{\gamma_{\alpha}+e^{i\varphi}}{\gamma_{\alpha}+e^{-i\varphi}}v\right)-f_{\alpha}\left(t,r,v\right)\right]\mathrm{sin\varphi}d\varphi.\label{eq:collision-integral-0=00003Dpi (w<0)}
\end{equation}
By combining Eq.\,(\ref{eq:collision-integral-0-pi(wc>0)}) and (\ref{eq:collision-integral-0=00003Dpi (w<0)}),
the collision integral can be expressed in a unified form as
\begin{equation}
\text{\ensuremath{\left(\frac{\partial f_{\alpha}}{\partial t}\right)}}_{c}^{\pm}=n_{D}(t,r)|v||r_{D}|\int_{0}^{\pi}\left[f_{\alpha}\left(t,r,\frac{\gamma_{\alpha}-s_{\alpha}e^{i\varphi}}{\gamma_{\alpha}-s_{\alpha}e^{-i\varphi}}v\right)-f_{\alpha}\left(t,r,v\right)\right]\mathrm{sin\varphi}d\varphi.\label{eq:collision-integral-unified}
\end{equation}
Here, the superscript $"+"$ corresponds to the case where $s_{\alpha}>0$,
while $"-"$ corresponds to the case where $s_{\alpha}<0$.

\subsection{Boltzmann equation and the moment equation}

With the collision integral derived, the Boltzmann equation can be
obtained as 
\begin{equation}
\frac{\partial f_{\alpha}}{\partial t}+\boldsymbol{v}\cdot\frac{\partial f_{\alpha}}{\partial\boldsymbol{r}}+\frac{\boldsymbol{F}_{\alpha}}{m_{\alpha}}\cdot\frac{\partial f_{\alpha}}{\partial\boldsymbol{v}}=\left(\frac{\partial f_{\alpha}}{\partial t}\right)_{c}^{\pm},\label{eq:Boltzmann-vector}
\end{equation}
where $\boldsymbol{F}_{\alpha}=\boldsymbol{F}_{\alpha}\left(t,\boldsymbol{r},\boldsymbol{v}\right)$
is the force field. Further details on this topic can be found in
the literature, such as in Ref.\,\citet{Lifshitz_1981}. In this
study, we assume that the force can be expressed as
\begin{equation}
\boldsymbol{F}_{\alpha}\left(t,\boldsymbol{r},\boldsymbol{v}\right)=\boldsymbol{F}_{\alpha}^{\left(1\right)}\left(t,\boldsymbol{r}\right)+\boldsymbol{F}_{\alpha}^{\left(2\right)}\left(t,\boldsymbol{r},\boldsymbol{v}\right),\label{eq:Force1+2}
\end{equation}
where $\boldsymbol{F}_{\alpha}^{\left(1\right)}$ is only dependent
on the position $\boldsymbol{r}$ and $\boldsymbol{F}_{\alpha}^{\left(2\right)}$
is perpendicular to the velocity, i.e., $\boldsymbol{F}_{\alpha}^{\left(2\right)}\cdot\boldsymbol{v}$.
In cases where no confusion arises, the complex form of collision
integral in Eq.\,(\ref{eq:Boltzmann-vector}) is automatically identified
as the vector form. Using Wirtinger derivative notation \citet{Remmert-Book-1991,Ken2009},
the Boltzmann equation (\ref{eq:Boltzmann-vector}) can also be expressed
in the form of complex variables as
\begin{equation}
\frac{\partial f_{\alpha}}{\partial t}+\frac{1}{2}\left[v\frac{\partial f_{\alpha}}{\partial r^{*}}+v^{*}\frac{\partial f_{\alpha}}{\partial r}\right]+\frac{1}{2}\left[\frac{F_{\alpha}}{m_{\alpha}}\frac{\partial f_{\alpha}}{\partial v^{*}}+\frac{F_{\alpha}^{*}}{m_{\alpha}}\frac{\partial f_{\alpha}}{\partial v}\right]=\left(\frac{\partial f_{\alpha}}{\partial t}\right)_{c}^{\pm}.\label{eq:Boltzmann-complex}
\end{equation}
From the microscopic description of the Boltzmann transport equation
(\ref{eq:Boltzmann-vector}) or (\ref{eq:Boltzmann-complex}), we
can transition to a fluid mechanics framework by averaging over velocity-dependent
functions. This approach, known as the moment method, yields a macroscopic
transport model \citet{Swanson-book-2008}.

Suppose $\psi\left(\boldsymbol{v}\right)$ is an arbitrary function
of velocity. Its average value, taken over the velocity space, is
defined as
\begin{equation}
\left\langle \psi\left(\boldsymbol{v}\right)\right\rangle _{\alpha}\equiv\frac{1}{n_{\alpha}\left(t,\boldsymbol{r}\right)}\int f_{\alpha}\left(t,\boldsymbol{r},\boldsymbol{v}\right)\psi\left(\boldsymbol{v}\right)d^{2}\boldsymbol{v},
\end{equation}
where 
\begin{equation}
n_{\alpha}\left(t,\boldsymbol{r}\right)=\int f_{\alpha}\left(t,\boldsymbol{r},\boldsymbol{v}\right)d^{2}\boldsymbol{v}
\end{equation}
is the local number density of particles of $\alpha$-species. By
multiplying Eq.\,(\ref{eq:Boltzmann-vector}) with $\psi\left(\boldsymbol{v}\right)$
and averaging the entire equation, we obtain

\begin{equation}
\frac{\partial}{\partial t}\left[n_{\alpha}\left\langle \psi\left(\boldsymbol{v}\right)\right\rangle _{\alpha}\right]+\nabla\cdot\left[n_{\alpha}\left\langle \psi\left(\boldsymbol{v}\right)\boldsymbol{v}\right\rangle _{\alpha}\right]-\frac{n_{\alpha}}{m_{\alpha}}\left\langle \boldsymbol{F}_{\alpha}\cdot\frac{\partial\psi}{\partial\boldsymbol{v}}\right\rangle _{\alpha}=\int\psi\left(\boldsymbol{v}\right)\left(\frac{\partial f_{\alpha}}{\partial t}\right)_{c}^{\pm}d^{2}\boldsymbol{v}.\label{eq:general-average-Eq}
\end{equation}
The detailed derivation of Eq.\,(\ref{eq:general-average-Eq}) can
be found in Ref.\,\citep{Swanson-book-2008}.

When $\psi\left(\boldsymbol{v}\right)=1$, the average
\begin{equation}
\left\langle \psi\left(\boldsymbol{v}\right)\right\rangle _{\alpha}=1,\label{eq:zeroth-moment}
\end{equation}
which represents the zeroth moment of the velocity. In this case,
the integral on the right-hand-side of Eq.\,(\ref{eq:general-average-Eq})
equals zero, i.e., 
\begin{equation}
\int\left(\frac{\partial f_{\alpha}}{\partial t}\right)_{c}^{\pm}d^{2}\boldsymbol{v}=0.\label{eq:collision-invariant-number}
\end{equation}
This result can be easily derived using Eq.\,(\ref{eq:collision-integral-unified})
and indicates that collisions do not change the total number of particles.
Substituting Eqs.\,(\ref{eq:zeroth-moment}) and (\ref{eq:collision-invariant-number})
into Eq.\,(\ref{eq:general-average-Eq}), we obtain the continuity
equation
\begin{equation}
\frac{\partial n_{\alpha}}{\partial t}+\nabla\cdot\left(n_{\alpha}\boldsymbol{V}_{\alpha}\right)=0.\label{eq:continuity-equation}
\end{equation}
where $\boldsymbol{V}_{\alpha}=\left\langle \boldsymbol{v}\right\rangle _{\alpha}$
is the average velocity of the fluid element (or fluid velocity) of
$\alpha$-species.

When $\psi\left(\boldsymbol{v}\right)=m_{\alpha}\boldsymbol{v}$,
the average value $\left\langle \psi\left(\boldsymbol{v}\right)\right\rangle _{\alpha}$
becomes
\begin{equation}
\left\langle m_{\alpha}\boldsymbol{v}\right\rangle _{\alpha}=m_{\alpha}\boldsymbol{V}_{\alpha}\label{eq:first-moment}
\end{equation}
which represents the first moment of the velocity. In this case, the
integral on the right-hand side of Eq.\,(\ref{eq:general-average-Eq})
represents the scattering force exerted by the background MDA, which
is given by 
\begin{equation}
\int m_{\alpha}\boldsymbol{v}\left(\frac{\partial f_{\alpha}}{\partial t}\right)_{c}^{\pm}d^{2}\boldsymbol{v}=\int m_{\alpha}\boldsymbol{w}_{\alpha}\left(\frac{\partial f_{\alpha}}{\partial t}\right)_{c}^{\pm}d^{2}\boldsymbol{v}\equiv\boldsymbol{R}_{\alpha}^{\pm}\left(t,\boldsymbol{r}\right),\label{eq:force-scattering}
\end{equation}
where $\boldsymbol{w}_{\alpha}=\boldsymbol{v}-\boldsymbol{V}_{\alpha}$
is the perturbation from the average velocity ($\left\langle \boldsymbol{w}_{\alpha}\right\rangle =0$).
Here, $\boldsymbol{R}_{\alpha}^{\pm}\left(t,\boldsymbol{r}\right)$
represents the momentum transferred from the background MDA to the
$\alpha$-species particles per unit time, or equivalently, the scattering
force exerted by the MDA. In deriving Eq.\,(\ref{eq:force-scattering}),
we used the relation
\begin{equation}
\int m_{\alpha}\boldsymbol{V}_{\alpha}\left(\frac{\partial f_{\alpha}}{\partial t}\right)_{c}^{\pm}d^{2}\boldsymbol{v}=\boldsymbol{V}_{\alpha}\left[\int m_{\alpha}\left(\frac{\partial f_{\alpha}}{\partial t}\right)_{c}^{\pm}d^{2}\boldsymbol{v}\right]=0.\label{eq:56}
\end{equation}
By directly integrating the left-hand side of Eq.\,(\ref{eq:force-scattering}),
it is straightforward to verify, using Eq.\,(\ref{eq:collision-integral-unified}),
that the scattering force vanishes in the absence of macroscopic flow,
i.e., 
\begin{equation}
\boldsymbol{R}_{\alpha}^{\pm}\left(t,\boldsymbol{r}\right)=0,\text{ when }\left\langle \boldsymbol{v}\right\rangle _{\alpha}=\boldsymbol{V}_{\alpha}=0.\label{eq:scattering-force-vanish}
\end{equation}
Substituting Eqs.\,(\ref{eq:first-moment}) and (\ref{eq:force-scattering})
into Eq.\,(\ref{eq:general-average-Eq}), we can obtain the momentum
equation
\begin{equation}
\frac{\partial}{\partial t}\left(m_{\alpha}n_{\alpha}\boldsymbol{V}_{\alpha}\right)+\nabla\cdot\left(m_{\alpha}n_{\alpha}\boldsymbol{V}_{\alpha}\boldsymbol{V}_{\alpha}\right)=-\nabla\cdot\boldsymbol{p}_{\alpha}+\boldsymbol{R}_{\alpha}^{\pm}+n_{\alpha}\boldsymbol{F}_{\alpha},\label{eq:momentum-equation}
\end{equation}
where $\boldsymbol{p}_{\alpha}\equiv m_{\alpha}n_{\alpha}\left\langle \boldsymbol{w}_{\alpha}\boldsymbol{w}_{\alpha}\right\rangle $
is the stress tensor and $\boldsymbol{F}_{\alpha}=\boldsymbol{F}_{\alpha}\left(t,\boldsymbol{r},\boldsymbol{V}_{\alpha}\right)$.

Let $\psi\left(\boldsymbol{v}\right)=m_{\alpha}\boldsymbol{v}^{2}/2$,
the average value $\left\langle \psi\left(\boldsymbol{v}\right)\right\rangle _{\alpha}$
becomes
\begin{equation}
\left\langle \frac{1}{2}m_{\alpha}\boldsymbol{v}^{2}\right\rangle _{\alpha}=\frac{1}{2}m_{\alpha}\boldsymbol{V}_{\alpha}^{2}+\frac{U_{\alpha}}{n_{\alpha}},\label{eq:second moment}
\end{equation}
which represents the second moment of the velocity. Here, $U_{\alpha}=m_{\alpha}n_{\alpha}\left\langle \boldsymbol{w}_{\alpha}^{2}\right\rangle $/2
is the internal energy density. Since collisions do not change the
energy, the integral on the right-hand side of Eq.\,(\ref{eq:general-average-Eq})
equals zero, i.e.,
\begin{equation}
\int\frac{1}{2}m_{\alpha}\boldsymbol{v}^{2}\left(\frac{\partial f_{\alpha}}{\partial t}\right)_{c}^{\pm}d^{2}\boldsymbol{v}=0.\label{eq:collision-invariant-energy}
\end{equation}
This result, which can also be derived directly from Eq.\,(\ref{eq:collision-integral-unified}),
implies that 
\begin{equation}
\boldsymbol{R}_{\alpha}^{\pm}\cdot\boldsymbol{V}_{\alpha}=0,\label{eq:R dot V =00003D0}
\end{equation}
meaning that the macroscopic force $\boldsymbol{R}_{\alpha}^{\pm}$
does no work on the fluid element. The detailed proof is provided
in Appendix.\,\ref{sec:proofs}. Substituting Eqs.\,(\ref{eq:second moment})
and (\ref{eq:collision-invariant-energy}) into Eq.\,(\ref{eq:general-average-Eq}),
we can obtain the energy equation
\begin{equation}
\frac{\partial U_{\alpha}}{\partial t}+\nabla\cdot\left(U_{\alpha}\boldsymbol{V}_{\alpha}\right)=-\boldsymbol{P}_{\alpha}:\nabla\boldsymbol{V}_{\alpha}+\nabla\cdot\boldsymbol{q}_{\alpha},\label{eq:energy-equation}
\end{equation}
where 
\begin{equation}
\boldsymbol{q}_{\alpha}=\frac{1}{2}m_{\alpha}n_{\alpha}\left\langle \boldsymbol{w}_{\alpha}^{2}\boldsymbol{w}_{\alpha}\right\rangle _{\alpha}\label{eq:heat-flux}
\end{equation}
is the heat flux density. The technique used to derive Eq.\,(\ref{eq:energy-equation})
can be found in Refs.\,\citep{Lifshitz_1981}and \citep{Swanson-book-2008}.
A detailed derivation is provided in Appendix.\,\ref{sec:proofs}.

\section{linear solution of Boltzmann equation\label{sec:linear solution of Boltzmann equation}}

\subsection{Collision integrals via Fourier decomposition}

We begin with Eqs.\,(\ref{eq:collision-integral-unified}) and (\ref{eq:Boltzmann-vector}).
Let 
\begin{equation}
v^{\pm}=\frac{\gamma_{\alpha}-s_{\alpha}e^{i\varphi}}{\gamma_{\alpha}-s_{\alpha}e^{-i\varphi}}v.\label{eq:v^(pm)}
\end{equation}
From Eqs.\,(\ref{eq:e^theta}) and (\ref{eq:inverse-velocity}),
we have $\left|v^{\pm}\right|=\left|v\right|$. Therefore, we can
rewrite $v$ and $v^{\pm}$ as 
\begin{align}
 & v=\left|v\right|e^{i\xi},\quad v^{\pm}=\left|v\right|e^{i\xi^{\pm}},\\
 & e^{i\xi^{\pm}}=\frac{\gamma_{\alpha}-s_{\alpha}e^{i\varphi}}{\gamma_{\alpha}-s_{\alpha}e^{-i\varphi}}e^{i\xi}.\label{eq:xi^(pm)-xi}
\end{align}
By rewriting the distribution function $f_{\alpha}\left(t,r,v\right)$
in terms of $|v|$ and $\xi$, $f_{\alpha}\left(t,r,v\right)$ becomes
a periodic function of $\xi$, i.e., 
\begin{equation}
f_{\alpha}\left(t,r,v\right)=f_{\alpha}\left(t,r,\left|v\right|,\xi\right)=f_{\alpha}\left(t,r,\left|v\right|,\xi+2\pi\right).\label{eq:31}
\end{equation}
 Hence, we can decompose $f_{\alpha}\left(t,r,v\right)$ with respect
to $\xi$ into a Fourier series
\begin{equation}
f_{\alpha}\left(t,r,v\right)=\sum_{m=-\infty}^{+\infty}f_{\alpha m}\left(t,r,\left|v\right|\right)e^{im\xi}.\label{eq:f_alpha-Fourier}
\end{equation}
Thus, $f_{\alpha}\left(t,r,v^{\pm}\right)$ becomes 
\begin{equation}
f_{\alpha}\left(t,r,v^{\pm}\right)=\sum_{m=-\infty}^{+\infty}f_{\alpha m}\left(t,r,\left|v\right|\right)e^{i\xi^{\pm}}=\sum_{m=-\infty}^{+\infty}f_{\alpha m}\left(t,r,\left|v\right|\right)\left(\frac{\gamma_{\alpha}-s_{\alpha}e^{i\varphi}}{\gamma_{\alpha}-s_{\alpha}e^{-i\varphi}}\right)^{m}e^{im\xi^{\pm}},\label{eq:f_alpha-Fourier-(pm)}
\end{equation}
where we used Eq.\,(\ref{eq:xi^(pm)-xi}). Substituting Eq.\,(\ref{eq:f_alpha-Fourier})
and Eq.\,(\ref{eq:f_alpha-Fourier-(pm)}) into Eq.\,(\ref{eq:collision-integral-unified}),
the collision integral becomes
\begin{equation}
\text{\ensuremath{\left(\frac{\partial f_{\alpha}}{\partial t}\right)}}_{c}^{\pm}=-n_{D}(t,r)\left|v\right|\left|r_{D}\right|\sum_{m\neq0}K_{m}^{\pm}\left(\gamma_{\alpha}\right)f_{\alpha m}\left(t,r,\left|v\right|\right)e^{im\xi},\label{eq:collision-integral-Fourier}
\end{equation}
where the function $K_{m}^{\pm}$ is defined as
\begin{equation}
K_{m}^{\pm}\left(\gamma_{\alpha}\right)\coloneqq\int_{0}^{\pi}\left[1-\left(\frac{\gamma_{\alpha}-s_{\alpha}e^{i\varphi}}{\gamma_{\alpha}-s_{\alpha}e^{-i\varphi}}\right)^{m}\right]\mathrm{sin\varphi}d\varphi.\label{eq:Km^(pm)}
\end{equation}
When $m=0$, it follows that 
\begin{equation}
K_{0}^{\pm}\left(\gamma_{\alpha}\right)\equiv0,\quad\forall\gamma_{\alpha}\in\mathbb{R},\label{eq:K0 =00003D 0}
\end{equation}
and $\forall m\in\mathbb{Z}$, it holds that
\begin{align}
\left[K_{m}^{\pm}\left(\gamma_{\alpha}\right)\right]^{*} & =K_{-m}^{\pm}\left(\gamma_{\alpha}\right),\label{eq:Km-K_m-conjugate-1}\\
\left[K_{m}^{+}\left(\gamma_{\alpha}\right)\right]^{*} & =K_{m}^{-}\left(\gamma_{\alpha}\right).\label{eq:Km-K_m-conjugate-2}
\end{align}
The derivation of Eqs.\,(\ref{eq:Km-K_m-conjugate-1}) and (\ref{eq:Km-K_m-conjugate-2})
is provided in the Appendix.\,\ref{sec:Km(x) function}.

\subsection{Calculation of the linear distribution function}

Assume that the distribution function can be expressed as 
\begin{equation}
f_{\alpha}=f_{\alpha0}+\delta f_{\alpha}\label{eq:f=00003Df0+f1}
\end{equation}
where $f_{\alpha0}$ is isotropic with respect to velocity and represents
the local equilibrium function, i.e., $f_{\alpha0}=f_{\alpha0}\left(r,\left|v\right|\right)$,
and $\delta f_{\alpha}\ll f_{\alpha0}$ represents a small deviation
from $f_{\alpha0}$. Substituting Eq.\,(\ref{eq:f=00003Df0+f1})
into Eq.\,(\ref{eq:Boltzmann-vector}), we obtain

\begin{align}
 & \frac{\partial\delta f_{\alpha}}{\partial t}+\left(\boldsymbol{v}\cdot\nabla f_{\alpha0}+\boldsymbol{v}\cdot\nabla\delta f_{\alpha}\right)+\left(\frac{\boldsymbol{F}_{\alpha}}{m_{\alpha}}\cdot\frac{\partial f_{\alpha0}}{\partial\boldsymbol{v}}+\frac{\boldsymbol{F}_{\alpha}}{m_{\alpha}}\cdot\frac{\partial\delta f_{\alpha}}{\partial\boldsymbol{v}}\right)\nonumber \\
 & =n_{D}(t,r)\left|v\right|\left|r_{D}\right|\int_{0}^{\pi}\left[\delta f_{\alpha}\left(t,r,v^{\pm}\right)-\delta f_{\alpha}\left(t,r,v\right)\right]\mathrm{sin\varphi}d\varphi,\label{eq:perturb-eq-full}
\end{align}
where we used the isotropic condition of $f_{\alpha0}$, which causes
the integral involving $f_{\alpha0}$ to vanish, i.e., 
\begin{equation}
\int_{0}^{\pi}\left[f_{\alpha0}\left(r,\left|v^{\pm}\right|\right)-f_{\alpha0}\left(r,\left|v\right|\right)\right]\mathrm{sin\varphi}d\varphi=0.\label{eq:39}
\end{equation}
Given that $\boldsymbol{v}\cdot\nabla f_{\alpha0}\gg\boldsymbol{v}\cdot\nabla\delta f_{\alpha}$
and $\left(\boldsymbol{F}_{\alpha}/m_{\alpha}\right)\cdot\left(\partial f_{0\alpha}/\partial\boldsymbol{v}\right)\gg\left(\boldsymbol{F}_{\alpha}/m_{\alpha}\right)\cdot\left(\partial\delta f_{\alpha}/\partial\boldsymbol{v}\right)$,
equation (\ref{eq:perturb-eq-full}) can be simplified by neglecting
the higher-order terms 
\begin{equation}
\frac{\partial\delta f_{\alpha}}{\partial t}+\boldsymbol{v}\cdot\nabla f_{0\alpha}+\frac{\boldsymbol{F}_{\alpha}}{m_{\alpha}}\cdot\frac{\partial f_{0\alpha}}{\partial\boldsymbol{v}}=n_{D}(t,r)\left|v\right|\left|r_{D}\right|\int_{0}^{\pi}\left[\delta f_{\alpha}\left(t,r,v^{\pm}\right)-\delta f_{\alpha}\left(t,r,v\right)\right]\mathrm{sin\varphi}d\varphi.\label{eq:perturb-eq-reduced}
\end{equation}

We now decompose $\delta f\left(t,r,v\right)$ into a Fourier series
\begin{equation}
\delta f_{\alpha}\left(t,r,v\right)=\sum_{m=-\infty}^{+\infty}\delta f_{\alpha m}\left(t,r,\left|v\right|\right)e^{im\xi}.\label{eq:f1-Fourier-series}
\end{equation}
Assuming that $\delta f\left(t,r,v\right)\propto e^{-i\omega t}\delta f\left(r,v\right)$,
we obtain
\begin{equation}
\frac{\partial\delta f_{\alpha}}{\partial t}=-i\omega\sum_{m=-\infty}^{+\infty}\delta f_{\alpha m}\left(t,r,\left|v\right|\right)e^{im\xi}.\label{eq:pa/pt-Fourier}
\end{equation}
Substituting Eqs.\,(\ref{eq:f1-Fourier-series}) and (\ref{eq:pa/pt-Fourier})
into Eq.\,(\ref{eq:perturb-eq-reduced}), we get
\begin{align}
 & -i\omega\sum_{m=-\infty}^{+\infty}\delta f_{\alpha m}\left(t,r,\left|v\right|\right)e^{im\xi}+\boldsymbol{v}\cdot\nabla f_{\alpha0}+\frac{\boldsymbol{F}_{\alpha}}{m_{\alpha}}\cdot\frac{\partial f_{\alpha0}}{\partial\boldsymbol{v}}\nonumber \\
 & =-n_{D}(t,r)\left|v\right|\left|r_{D}\right|\sum_{m\neq0}K_{m}^{\pm}\left(\gamma_{\alpha}\right)\delta f_{\alpha m}\left(t,r,\left|v\right|\right)e^{im\xi},\label{eq:perturb-eq-reduced-Fouriers-s1}
\end{align}
which can be equivalently rewritten as
\begin{equation}
\boldsymbol{v}\cdot\nabla f_{\alpha0}+\frac{\boldsymbol{F}_{\alpha}}{m_{\alpha}}\cdot\frac{\partial f_{0\alpha}}{\partial\boldsymbol{v}}=\sum_{m=-\infty}^{+\infty}-\left[\nu_{\alpha m}^{\pm}\left(t,r,\left|v\right|\right)-i\omega\right]\delta f_{\alpha m}\left(t,r,\left|v\right|\right)e^{im\xi},\label{eq:perturb-eq-reduced-Fouriers-omega}
\end{equation}
where 
\begin{equation}
\nu_{\alpha m}^{\pm}\left(t,r,\left|v\right|\right)=n_{D}(t,r)\left|v\right||r_{D}|K_{m}^{\pm}\left(\gamma_{\alpha}\right).\label{eq:nu_m(pm)}
\end{equation}
The Fourier coefficient $\delta f_{\alpha m}$ of $\delta f_{\alpha}$
can be determined by calculating the Fourier coefficients of Eq.\,(\ref{eq:perturb-eq-reduced-Fouriers-omega})
\begin{equation}
-2\pi\left[\nu_{\alpha m}^{\pm}\left(t,r,\left|v\right|\right)-i\omega\right]\delta f_{\alpha m}=\left(e^{im\xi},\boldsymbol{v}\cdot\nabla f_{\alpha0}+\frac{\boldsymbol{F}_{\alpha}}{m_{\alpha}}\cdot\frac{\partial f_{0\alpha}}{\partial\boldsymbol{v}}\right).\label{eq:inner-product-F0}
\end{equation}
Here, we have used the inner product notation $\left(\:,\:\right)$,
defined as
\begin{equation}
\left(h\left(\xi\right),g\left(\xi\right)\right):=\int_{0}^{2\pi}h^{*}\left(\xi\right)g\left(\xi\right)d\xi,
\end{equation}
where $h\left(\xi\right)$ and $g\left(\xi\right)$ are complex functions
of $\xi$. We also used the fact that
\begin{equation}
\left(e^{im\xi},e^{in\xi}\right)=\int_{0}^{2\pi}e^{-im\xi}e^{in\xi}d\xi=2\pi\delta_{mn},\label{eq:Orthogonal-Fourier}
\end{equation}
where $\delta_{mn}$ is Kronecker delta, which equals $1$ if 1 if
$m=n$ and $0$ otherwise. The right-hand side of Eq.\,(\ref{eq:inner-product-F0})
can be transformed into
\begin{align}
 & \left(e^{im\xi},\boldsymbol{v}\cdot\nabla f_{\alpha0}+\frac{\boldsymbol{F}_{\alpha}}{m_{\alpha}}\cdot\frac{\partial f_{0\alpha}}{\partial\boldsymbol{v}}\right)=\left(e^{im\xi},\left|\boldsymbol{v}\right|\boldsymbol{e}_{v}\cdot\nabla f_{\alpha0}+\frac{\partial f_{\alpha0}}{\partial\left|\boldsymbol{v}\right|}\boldsymbol{e}_{v}+\frac{1}{\left|\boldsymbol{v}\right|}\frac{\partial f_{\alpha0}}{\partial\xi}\boldsymbol{e}_{\xi}\right)\nonumber \\
 & =\left|\boldsymbol{v}\right|\nabla f_{\alpha0}\cdot\left(e^{im\xi},\boldsymbol{e}_{v}\right)+\frac{\partial f_{\alpha0}}{\partial\left|\boldsymbol{v}\right|}\frac{\boldsymbol{F}_{\alpha}}{m_{\alpha}}\cdot\left(e^{im\xi},\boldsymbol{e}_{v}\right)=\left[\left|\boldsymbol{v}\right|\nabla f_{\alpha0}+\frac{\boldsymbol{F}_{\alpha}}{m_{\alpha}}\frac{\partial f_{\alpha0}}{\partial\left|\boldsymbol{v}\right|}\right]\cdot\left(e^{im\xi},\boldsymbol{e}_{v}\right)\nonumber \\
 & =\left(\hat{\boldsymbol{\mathcal{D}}}_{\alpha}f_{\alpha0}\right)\cdot\left(e^{im\xi},\boldsymbol{e}_{v}\right),\label{eq:50}
\end{align}
where the operator $\hat{\boldsymbol{\mathcal{D}}}_{\alpha}$ is defined
as
\begin{equation}
\hat{\boldsymbol{\mathcal{D}}}_{\alpha}\equiv\left|\boldsymbol{v}\right|\frac{\partial}{\partial\boldsymbol{r}}+\frac{\boldsymbol{F}_{\alpha}}{m_{\alpha}}\frac{\partial}{\partial\left|\boldsymbol{v}\right|}.\label{eq:D-operator}
\end{equation}
Here, $\boldsymbol{e}_{v}$ and $\boldsymbol{e}_{\xi}$ are unit vectors
in the radial and azimuthal directions, respectively, in the polar
coordinate system of velocity space, defined as
\begin{align}
\boldsymbol{e}_{v} & =\boldsymbol{v}/\left|\boldsymbol{v}\right|=\mathrm{cos\xi}\boldsymbol{e}_{x}+\mathrm{sin}\xi\boldsymbol{e}_{y},\\
\boldsymbol{e}_{\xi} & =-\mathrm{sin}\xi\boldsymbol{e}_{x}+\mathrm{cos\xi}\boldsymbol{e}_{y}
\end{align}
In deriving Eq.\,(\ref{eq:50}) we used $\partial f_{\alpha0}/\partial\xi=0$
due to the isotropy of $f_{\alpha0}$ and $\boldsymbol{F}_{\alpha}\cdot\boldsymbol{e}_{v}=0$
as dictated by Eq.\,(\ref{eq:Force1+2}). The last term on the right-hand
side of Eq.\,(\ref{eq:50}) can be directly integrated as 
\begin{equation}
\left(e^{im\xi},\boldsymbol{e}_{v}\right)=\begin{cases}
\pi\left(\boldsymbol{e}_{x}-i\boldsymbol{e}_{y}\right), & m=1,\\
\pi\left(\boldsymbol{e}_{x}+i\boldsymbol{e}_{y}\right), & m=-1,\\
0, & m\neq\pm1.
\end{cases}\label{eq:89}
\end{equation}
Substituting Eq.\,(\ref{eq:50}) into Eq.\,(\ref{eq:50}), and combining
the results with Eq.\,(\ref{eq:inner-product-F0}), we can obtain
\begin{equation}
\delta f_{\alpha m}=\begin{cases}
-\frac{\hat{\boldsymbol{\mathcal{D}}}_{\alpha}f_{0}\cdot\left(\boldsymbol{e}_{x}-i\boldsymbol{e}_{y}\right)}{2\left[\nu_{\alpha1}^{\pm}\left(t,r,\left|v\right|\right)-i\omega\right]}, & m=1,\\
-\frac{\hat{\boldsymbol{\mathcal{D}}}_{\alpha}f_{0}\cdot\left(\boldsymbol{e}_{x}+i\boldsymbol{e}_{y}\right)}{2\left[\nu_{\alpha-1}^{\pm}\left(t,r,\left|v\right|\right)-i\omega\right]}, & m=-1,\\
0, & m\neq\pm1.
\end{cases}\label{eq:delta_f_m}
\end{equation}
The perturbed distribution function $\delta f_{\alpha}$ can be derived
by substituting Eq.\,(\ref{eq:delta_f_m}) into Eq.\,(\ref{eq:f1-Fourier-series})
as 
\begin{align}
 & \delta f_{\alpha}\left(t,r,v\right)=\delta f_{\alpha-1}e^{-i\xi}+\delta f_{\alpha1}e^{i\xi}\nonumber \\
 & =-\frac{\hat{\boldsymbol{\mathcal{D}}}_{\alpha}f_{0}\cdot\left(\boldsymbol{e}_{x}+i\boldsymbol{e}_{y}\right)}{2\left[\nu_{\alpha-1}^{\pm}-i\omega\right]}e^{-i\xi}-\frac{\hat{\boldsymbol{\mathcal{D}}}_{\alpha}f_{0}\cdot\left(\boldsymbol{e}_{x}-i\boldsymbol{e}_{y}\right)}{2\left[\nu_{\alpha1}^{\pm}\left(t,r,\left|v\right|\right)-i\omega\right]}e^{i\xi}\nonumber \\
 & -\frac{\hat{\boldsymbol{\mathcal{D}}}_{\alpha}f_{0}\cdot\boldsymbol{e}_{x}e^{-i\xi}+i\hat{\boldsymbol{\mathcal{D}}}_{\alpha}f_{0}\cdot\boldsymbol{e}_{y}e^{-i\xi}}{2\left[\left(\nu_{\alpha1}^{\pm}\right)^{*}-i\omega\right]}-\frac{\hat{\boldsymbol{\mathcal{D}}}_{\alpha}f_{0}\cdot\boldsymbol{e}_{x}e^{i\xi}-i\hat{\boldsymbol{\mathcal{D}}}_{\alpha}f_{0}\cdot\boldsymbol{e}_{y}e^{i\xi}}{2\left[\nu_{\alpha1}^{\pm}-i\omega\right]}\nonumber \\
 & =-\frac{1}{2}\left[\frac{e^{i\xi}}{\nu_{\alpha1}^{\pm}-i\omega}+\frac{e^{-i\xi}}{\left(\nu_{\alpha1}^{\pm}\right)^{*}-i\omega}\right]\hat{\boldsymbol{\mathcal{D}}}_{\alpha}f_{0}\cdot\boldsymbol{e}_{x}-\frac{1}{2i}\left[\frac{e^{i\xi}}{\nu_{\alpha1}^{\pm}-i\omega}-\frac{e^{-i\xi}}{\left(\nu_{\alpha1}^{\pm}\right)^{*}-i\omega}\right]\hat{\boldsymbol{\mathcal{D}}}_{\alpha}f_{0}\cdot\boldsymbol{e}_{y}.\label{eq:91}
\end{align}

In low-frequency condition, where $\omega\approx0$ (or $\partial\delta f_{\alpha0}/\partial t\approx0$,
i.e., $\delta f_{\alpha}=\delta f_{\alpha}\left(r,v\right)$), and
$\nu_{\alpha1}^{\pm}=\nu_{\alpha1}^{\pm}\left(r,\left|v\right|\right)$
is time-independent, equation (\ref{eq:91}) simplifies to 
\begin{align}
 & \delta f_{\alpha}\left(r,v\right)=-\frac{1}{2}\left[\frac{e^{i\xi}}{\nu_{\alpha1}^{\pm}}+\frac{e^{-i\xi}}{\left(\nu_{\alpha1}^{\pm}\right)^{*}}\right]\hat{\boldsymbol{\mathcal{D}}}_{\alpha}f_{0}\cdot\boldsymbol{e}_{x}-\frac{1}{2i}\left[\frac{e^{i\xi}}{\nu_{\alpha1}^{\pm}}-\frac{e^{-i\xi}}{\left(\nu_{\alpha1}^{\pm}\right)^{*}}\right]\hat{\boldsymbol{\mathcal{D}}}_{\alpha}f_{0}\cdot\boldsymbol{e}_{y}\nonumber \\
 & =-\hat{\boldsymbol{\mathcal{D}}}_{\alpha}f_{0}\cdot\left[\mathrm{Re}\left(\frac{e^{i\xi}}{\nu_{\alpha1}^{\pm}}\right)\boldsymbol{e}_{x}+\mathrm{Im}\left(\frac{e^{i\xi}}{\nu_{\alpha1}^{\pm}}\right)\boldsymbol{e}_{y}\right].\label{eq:92}
\end{align}
By defining
\begin{equation}
\frac{1}{\nu_{\alpha1}^{\pm}\left(r,\left|v\right|\right)}\coloneqq\tau_{\alpha1}\left(r,\left|v\right|\right)\mp i\thinspace\tau_{\alpha2}\left(r,\left|v\right|\right),\label{eq:tau1+tau2}
\end{equation}
and substituting this expression into Eq.\,(\ref{eq:92}), we can
rewrite it in vector form as
\begin{align}
 & \delta f_{\alpha}\left(r,v\right)\nonumber \\
 & =-\hat{\boldsymbol{\mathcal{D}}}_{\alpha}f_{0}\cdot\left\{ \left[\tau_{\alpha1}\mathrm{cos\xi}-\left(\mp\tau_{\alpha2}\right)\mathrm{sin}\xi\right]\boldsymbol{e}_{x}+\left[\tau_{\alpha1}\mathrm{sin}\xi\mp\tau_{\alpha2}\mathrm{cos\xi}\right]\boldsymbol{e}_{y}\right\} \nonumber \\
 & =-\left[\tau_{\alpha1}\mathrm{cos\xi}-\left(\mp\tau_{\alpha2}\right)\mathrm{sin}\xi\right]\hat{\mathcal{D}_{x}}f_{0}+\left[\tau_{\alpha1}\mathrm{sin}\xi\mp\tau_{\alpha2}\mathrm{cos\xi}\right]\hat{\mathcal{D}_{y}}f_{0}\nonumber \\
 & =-\left(\mathrm{cos\xi}\quad\mathrm{sin}\xi\right)\left(\begin{array}{cc}
\tau_{\alpha1} & \mp\tau_{\alpha2}\\
\pm\tau_{\alpha2} & \tau_{\alpha1}
\end{array}\right)\left(\begin{array}{c}
\hat{\mathcal{D}_{x}}f_{0}\\
\hat{\mathcal{D}_{y}}f_{0}
\end{array}\right).\label{eq:94}
\end{align}
To distinguish the perturbed distribution function $\delta f_{\alpha}$
arising from different types of collisions, we denote it as $\delta f_{\alpha}^{\pm}$.
Thus, Eq.\,(\ref{eq:94}) can be rewritten as 
\begin{equation}
\delta f_{\alpha}^{\pm}\left(r,v\right)=-\boldsymbol{e}_{v}\cdot\boldsymbol{M}_{\alpha}^{\pm}\left(r,\left|v\right|\right)\cdot\hat{\boldsymbol{\mathcal{D}}}_{\alpha}f_{\alpha0},\label{eq:delta-f-finally}
\end{equation}
where we define
\begin{equation}
\boldsymbol{M}_{\alpha}^{\pm}\left(r,\left|v\right|\right)=\left(\begin{array}{cc}
\tau_{\alpha1} & \mp\tau_{\alpha2}\\
\pm\tau_{\alpha2} & \tau_{\alpha1}
\end{array}\right).\label{eq:M+-}
\end{equation}
From Eq.\,(\ref{eq:M+-}), it can be observed that 
\begin{equation}
M_{\alpha ij}^{\pm}=M_{\alpha ji}^{\pm},\text{ for }i\ne j.\label{eq:M-off-diagnal}
\end{equation}

\section{typical transport coefficients\label{sec:typical-transport-coefficient}}

We start with Eq.\,(\ref{eq:delta-f-finally}) to calculate the transport
coefficient, applying the vector form of all functions in the following
context. Assuming that the local equilibrium is described by a mono-kinetic
distribution, we have

\begin{equation}
f_{\alpha0}\left(\boldsymbol{r},\left|\boldsymbol{v}\right|\right)=\frac{n_{\alpha0}\left(\boldsymbol{r}\right)}{2\pi\left|\boldsymbol{v}\right|}\delta\left(\left|\boldsymbol{v}\right|-u_{\alpha}\left(\boldsymbol{r}\right)\right)=\frac{n_{\alpha0}\left(\boldsymbol{r}\right)}{2\pi u_{\alpha}\left(\boldsymbol{r}\right)}\delta\left(\left|\boldsymbol{v}\right|-u_{\alpha}\left(\boldsymbol{r}\right)\right),\label{eq:mono-distribution}
\end{equation}
indicating that all particles have the same speed $u\left(\boldsymbol{r}\right)$
at position $\boldsymbol{r}$, but with different directions.. It
is straightforward to verify that $\int f_{\alpha0}\left(\boldsymbol{r},\left|\boldsymbol{v}\right|\right)d^{2}\boldsymbol{v}=n_{\alpha0}\left(\boldsymbol{r}\right)$.

\subsection{Particle transport and diffusion tensor \label{subsec:Particle-transport}}

Suppose the external force field $\boldsymbol{F}_{\alpha}=0$, and
the particle speed is constant, i.e., $u_{\alpha}\left(\boldsymbol{r}\right)=u_{\alpha}$,
but the particle density depends on $\boldsymbol{r}$. Under these
conditions, the operator simplifies to $\hat{\boldsymbol{\mathcal{D}}}=\left|\boldsymbol{v}\right|\partial/\partial\boldsymbol{r}$.
Using Eq.\,(\ref{eq:delta-f-finally}), the perturbed distribution
function $\delta f_{\alpha}^{\pm}$ is then given by
\begin{equation}
\delta f_{\alpha}^{\pm}=-\frac{1}{2\pi}\delta\left(\left|\boldsymbol{v}\right|-u_{\alpha}\right)\boldsymbol{e}_{v}\cdot\boldsymbol{M}_{\alpha}^{\pm}\cdot\nabla n_{\alpha0}\left(\boldsymbol{r}\right).\label{eq:delta-f-diffusion}
\end{equation}
The particle flux of the $\alpha$-species is calculated as
\begin{align}
\boldsymbol{\Gamma}_{\alpha}^{\pm}\left(\boldsymbol{r}\right) & =n_{\alpha}\left(\boldsymbol{r}\right)\boldsymbol{V}_{\alpha}=\int f\boldsymbol{v}d^{2}\boldsymbol{v}\nonumber \\
 & =\int\left(f_{0\alpha}+\delta f_{\alpha}^{\pm}\right)\boldsymbol{v}d^{2}\boldsymbol{v}=\int\delta f_{\alpha}^{\pm}\left|\boldsymbol{v}\right|\boldsymbol{e}_{v}d\boldsymbol{v}\nonumber \\
 & =-\frac{1}{2\pi}\left[\int\left|\boldsymbol{v}\right|\delta\left(\left|\boldsymbol{v}\right|-u_{\alpha}\right)\boldsymbol{e}_{v}\boldsymbol{e}_{v}\cdot\boldsymbol{M}_{\alpha}^{\pm}\left(\boldsymbol{r},\left|\boldsymbol{v}\right|\right)d^{2}\boldsymbol{v}\right]\cdot\nabla n_{\alpha0}\left(\boldsymbol{r}\right)\nonumber \\
 & =-\frac{1}{2\pi}\left\{ \left[\int_{0}^{2\pi}\boldsymbol{e}_{v}\boldsymbol{e}_{v}d\xi\right]\cdot\left[\int_{0}^{+\infty}\boldsymbol{M}_{\alpha}^{\pm}\left(\boldsymbol{r},\left|\boldsymbol{v}\right|\right)\left|\boldsymbol{v}\right|^{2}\delta\left(\left|\boldsymbol{v}\right|-u_{\alpha}\right)d\left|\boldsymbol{v}\right|\right]\right\} \cdot\nabla n_{\alpha0}\left(\boldsymbol{r}\right)\nonumber \\
 & =-\frac{u^{2}}{2}\boldsymbol{M}_{\alpha}^{\pm}\left(\boldsymbol{r},u_{\alpha}\right)\cdot\nabla n_{\alpha0}\left(\boldsymbol{r}\right),\label{eq:particle-flux-1}
\end{align}
where we used 
\begin{align}
 & \int_{0}^{2\pi}\boldsymbol{e}_{v}\boldsymbol{e}_{v}d\xi=\int_{0}^{2\pi}\left(\mathrm{cos\xi}\boldsymbol{e}_{x}+\mathrm{sin}\xi\boldsymbol{e}_{y}\right)\left(\mathrm{cos\xi}\boldsymbol{e}_{x}+\mathrm{sin}\xi\boldsymbol{e}_{y}\right)d\xi\nonumber \\
 & =\int_{0}^{2\pi}\left[\mathrm{cos^{2}\xi}\boldsymbol{e}_{x}\boldsymbol{e}_{x}+\mathrm{cos\xi}\mathrm{sin}\xi\left(\boldsymbol{e}_{x}\boldsymbol{e}_{y}+\boldsymbol{e}_{y}\boldsymbol{e}_{x}\right)\mathrm{sin}^{2}\xi\boldsymbol{e}_{y}\boldsymbol{e}_{y}\right]d\xi\nonumber \\
 & =\pi\left(\boldsymbol{e}_{x}\boldsymbol{e}_{x}+\boldsymbol{e}_{y}\boldsymbol{e}_{y}\right)=\pi\boldsymbol{I},\label{eq:59}
\end{align}
and \textbf{$\boldsymbol{I}$} is the 2D identity tensor.

According to the equipartition theorem of energy $T_{\alpha}/2=\left\langle \frac{1}{2}m_{\alpha}v_{x}^{2}\right\rangle _{\alpha}=\left\langle \frac{1}{2}m_{\alpha}v_{y}^{2}\right\rangle _{\alpha}$,
where $T_{\alpha}$ is the temperature of the $\alpha$-species. We
have set the Boltzmann constant to 1, i.e., $k_{B}=1$, yielding
\begin{equation}
\frac{u_{\alpha}^{2}}{2}=\left\langle \frac{1}{2}\left(v_{x}^{2}+v_{y}^{2}\right)\right\rangle _{\alpha}=\left\langle \frac{\left|\boldsymbol{v}\right|^{2}}{2}\right\rangle _{\alpha}=\frac{T_{\alpha}}{m_{\alpha}}.\label{eq:T/m}
\end{equation}
Substituting Eq.\,(\ref{eq:T/m}) into Eq.\,(\ref{eq:particle-flux-1}),
we finally obtain
\begin{equation}
\boldsymbol{\Gamma}_{\alpha}^{\pm}\left(\boldsymbol{r}\right)=-\boldsymbol{D}_{\alpha}^{\pm}\left(\boldsymbol{r},u_{\alpha}\right)\cdot\nabla n_{\alpha0}\left(\boldsymbol{r}\right),\label{eq:particle-flux-2}
\end{equation}
where
\begin{equation}
\boldsymbol{D}_{\alpha}^{\pm}\left(\boldsymbol{r},u_{\alpha}\right)=\frac{T_{\alpha}}{m_{\alpha}}\boldsymbol{M}_{\alpha}^{\pm}\left(\boldsymbol{r},u_{\alpha}\right)\label{eq:diffusion tensor}
\end{equation}
is the Hall (or odd) diffusion tensor. Its off-diagonal elements are
antisymmetric due to Eq.\,(\ref{eq:M-off-diagnal}).

Since $\boldsymbol{M}_{\alpha}^{\pm}$ is related to the direction
of the magnetic field within the magnetic region (see Eqs.\,(\ref{eq:Km^(pm)}),
(\ref{eq:nu_m(pm)}), (\ref{eq:tau1+tau2}) and (\ref{eq:M+-})),
when the magnetic field in the MDA changes direction, the gyrofrequency
$\omega_{c\alpha}$ will change its sign, and $\nu_{\alpha1}^{\pm}\left(r,\left|v\right|\right)$
will transform into their conjugates, i.e., $\nu_{\alpha1}^{\pm}\rightarrow\left(\nu_{\alpha1}^{\pm}\right)^{*}$
(or $\tau_{\alpha2}\rightarrow-\tau_{\alpha2}$). As a result, we
have
\begin{equation}
M_{\alpha ij}^{\pm}\left[\boldsymbol{B}\right]=M_{\alpha ji}^{\mp}\left[-\boldsymbol{B}\right].
\end{equation}
Consequently, the diffusion tensor will satisfy the relation
\begin{equation}
D_{\alpha ij}^{\pm}\left[\boldsymbol{B}\right]=D_{\alpha ji}^{\mp}\left[-\boldsymbol{B}\right],\label{eq:107}
\end{equation}
which is consistent with Onsager's reciprocal relations for kinetic
coefficients.

Assume that $n_{\alpha0}$ varies slowly over time, ensuring that
each state is in a quasi-stationary condition, and Eq.\,(\ref{eq:particle-flux-2})
continues to hold. Substituting Eq.\,(\ref{eq:particle-flux-2})
into Eq.\,(\ref{eq:continuity-equation}), we obtain the Hall diffusion
equation

\begin{equation}
\frac{\partial n_{\alpha0}}{\partial t}-\nabla\cdot\left[\boldsymbol{D}_{\alpha}^{\pm}\left(\boldsymbol{r},u_{\alpha}\right)\cdot\nabla n_{\alpha0}\right]=0,\label{eq:diffusion-equation-0}
\end{equation}
where the perturbed density $\delta n_{\alpha}=n_{\alpha}-n_{\alpha0}$
is neglected. For stationary and uniform MDA (or $n_{D}$ is constant),
$\tau_{\alpha1}$ and $\tau_{\alpha2}$ is then independent of $\boldsymbol{r}$,
consequently, the diffusion tensor is a constant tensor. Equation
(\ref{eq:diffusion-equation-0}) is then reduced to
\begin{equation}
\frac{\partial n_{\alpha0}}{\partial t}-\boldsymbol{D}_{\alpha}^{\pm}:\nabla\nabla n_{\alpha0}=0,\label{eq:diffusion-equation-1}
\end{equation}
Substituting Eqs.\,(\ref{eq:M+-}) and (\ref{eq:diffusion tensor})
into Eq.\,(\ref{eq:diffusion-equation-1}), Eq.\,(\ref{eq:diffusion-equation-1})
can be equivalently transformed into
\begin{equation}
\frac{\partial n_{\alpha0}}{\partial t}-\frac{T_{\alpha}\tau_{\alpha1}}{m_{\alpha}}\nabla^{2}n_{\alpha0}=0.\label{eq:diffusion-equation-2}
\end{equation}

To simulate the diffusion properties of particles using Eq.\,(\ref{eq:diffusion-equation-1}),
it is convenient to nondimensionalize the problem using the following
substitutions:
\begin{align}
 & \tilde{n}_{\alpha}=\left(\pi\left|r_{D}\right|^{2}\right)n_{\alpha},\label{eq:n=00005Ctilde}\\
 & \tilde{t}=\left|\omega_{c\alpha}\right|t,\;\tilde{\boldsymbol{r}}=\frac{\boldsymbol{r}}{\left|r_{D}\right|}.\label{eq:t=00005Ctilde+x=00005Ctilde}
\end{align}
Substituting these nondimensional quantities into Eq.\,(\ref{eq:diffusion-equation-1}),
we have
\begin{equation}
\frac{\partial\tilde{n}_{\alpha0}}{\partial\tilde{t}}-\tilde{\boldsymbol{D}}_{\alpha}^{\pm}:\tilde{\nabla}\tilde{\nabla}\tilde{n}_{\alpha0}=0,
\end{equation}
where $\tilde{\nabla}\equiv\partial/\partial\tilde{\boldsymbol{x}}$
and
\begin{equation}
\tilde{\boldsymbol{D}}_{\alpha}^{\pm}=\frac{\boldsymbol{D}_{\alpha}^{\pm}}{\left|\omega_{c\alpha}\right|\left|r_{D}\right|^{2}}=\frac{1}{2}\gamma_{\alpha}\left|\omega_{c\alpha}\right|\left(\begin{array}{cc}
\tau_{\alpha1} & \mp\tau_{\alpha2}\\
\pm\tau_{\alpha2} & \tau_{\alpha1}
\end{array}\right)
\end{equation}
nondimensionalized diffusion tensor. In this study, we simulate the
behavior of diffusion equations using the COMSOL Multiphysics software.
We examine three distinct types of diffusion equations, each corresponding
to a different diffusion tensor $\tilde{\boldsymbol{D}}_{\alpha}^{-}$,
$\tilde{\boldsymbol{D}}_{\alpha}^{+}$ and $\tilde{\boldsymbol{D}}_{\alpha}^{0}$,
the simulation results are shown in Fig.\,\ref{fig:particle-diffusion}(a-c),
(d-f), and (g-i). Here, $\tilde{\boldsymbol{D}}_{\alpha}^{0}$ denotes
the normal diffusion tensor, which is obtained by setting $\tau_{\alpha2}=0$
in $\tilde{\boldsymbol{D}}_{\alpha}^{\pm}$ . This results in
\begin{equation}
\tilde{\boldsymbol{D}}_{\alpha}^{0}=\frac{1}{2}\gamma_{\alpha}\left|\omega_{c\alpha}\right|\left(\begin{array}{cc}
\tau_{\alpha1} & 0\\
0 & \tau_{\alpha1}
\end{array}\right).
\end{equation}
The simulations are performed within a $1\times2$ rectangular domain.
We set the parameters as $\tau_{\alpha1}=2/\gamma_{\alpha}\left|\omega_{c\alpha}\right|$
and $\tau_{\alpha2}=\pi/\left(2\gamma_{\alpha}\left|\omega_{c\alpha}\right|\right)$.
Under these conditions, the diffusion tensors are defined as
\begin{equation}
\tilde{\boldsymbol{D}}_{\alpha}^{+}=\left(\begin{array}{cc}
1 & -\frac{\pi}{4}\\
\frac{\pi}{4} & 1
\end{array}\right),\:\tilde{\boldsymbol{D}}_{\alpha}^{-}=\left(\begin{array}{cc}
1 & \frac{\pi}{4}\\
-\frac{\pi}{4} & 1
\end{array}\right),\:\tilde{\boldsymbol{D}}_{\alpha}^{0}=\left(\begin{array}{cc}
1 & 0\\
0 & 1
\end{array}\right).
\end{equation}
The boundary conditions are specified as follows: at the top and bottom
boundaries, a flux condition is imposed with $\boldsymbol{\Gamma}_{\alpha}^{\pm}\text{ and }\boldsymbol{\Gamma}_{\alpha}^{0}=0$,
ensuring that there is no net particle flux across these boundaries,
where $\boldsymbol{\Gamma}_{\alpha}^{0}=-\tilde{\boldsymbol{D}}_{\alpha}^{0}\cdot\nabla n_{\alpha0}$.
Dirichlet boundary conditions are applied at the left and right boundaries
to fix the particle density: at the left boundary $n\left(t,x=0,y\right)=n_{0}$
with $n_{0}=10$, and at the right boundary $n\left(t,x=2,y\right)=0$.
The initial particle density is given by a Gaussian distribution $n\left(t=0,x,y\right)=n_{0}e^{-x^{2}/L^{2}}$,
where $L=0.08$.

Figure \ref{fig:one-way preferential diffusion}(a) and (c) show the
simulation results for the stationary equation of electrons in a rectangular
region consisting of two 1\texttimes 4 rectangles, with boundary conditions
identical to those in Fig.\,\ref{fig:particle-diffusion}. We observe
that the electrons exhibit one-way preferential diffusion. The results
indicate that the magnetic field asymmetry plays a crucial role in
controlling electron diffusion. When the magnetic field is oriented
with $B_{z}>0$ in one domain and $B_{z}<0$ in the other, the diffusion
behavior becomes highly directional, leading to one-way preferential
diffusion. This can be attributed to the antisymmetric nature of the
Hall diffusion tensor, which introduces an effective barrier to electron
movement in one direction, while facilitating it in the other.

\begin{figure}
\includegraphics[scale=0.7]{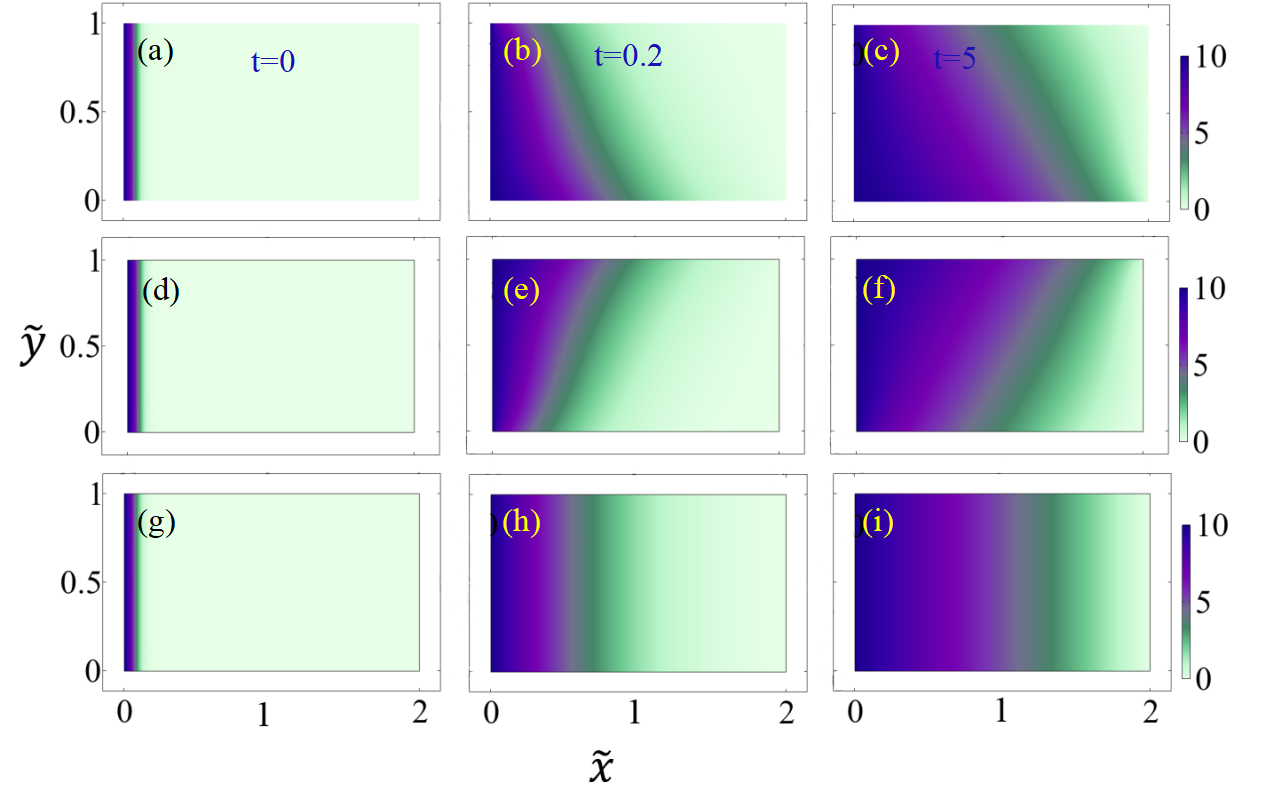}

\caption{\label{fig:particle-diffusion} Evolution of particle density over
time for three different diffusion tensors. Particles diffuse from
left to right. (a)-(c) show the evolution of particle density over
time for the diffusion tensor $\tilde{\boldsymbol{D}}_{\alpha}^{-}$,
where particles tend to exhibit bottom-biased motion. (d)-(f) depict
the diffusion tensor $\tilde{\boldsymbol{D}}_{\alpha}^{+}$, where
particles tend to exhibit top-biased motion. (g)-(i) illustrate the
diffusion tensor $\tilde{\boldsymbol{D}}_{\alpha}^{0}$, where particles
diffuse from left to right with no preferential direction toward the
top or bottom.}
\end{figure}

\begin{figure}
\includegraphics[scale=0.3]{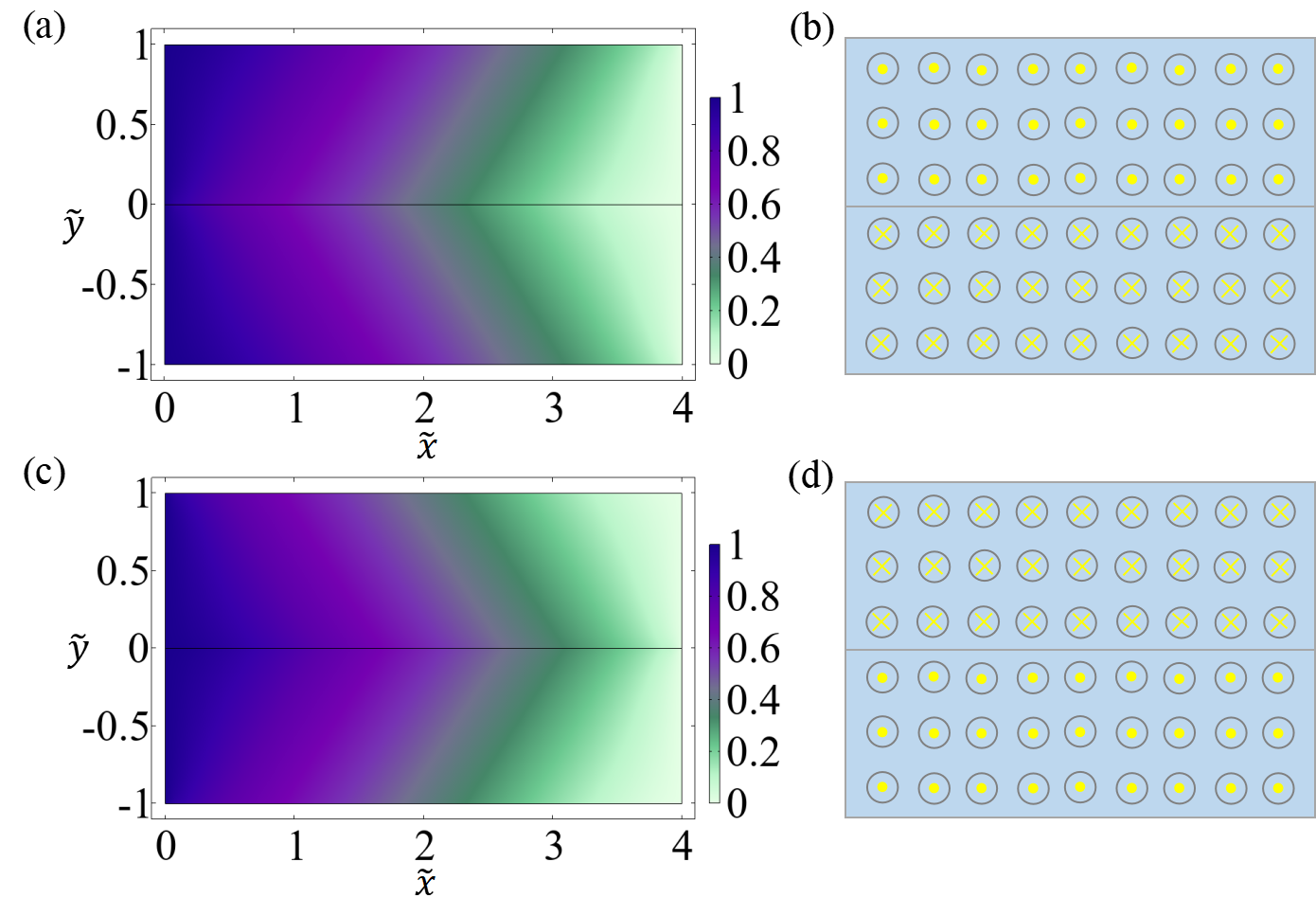}

\caption{\label{fig:one-way preferential diffusion} Electrons exhibit a one-way
preferential diffusion phenomenon at the interface of two MDAs. (a)
When electrons diffuse within the MDA shown in (b), where $B_{z}>0$
at the top and $B_{z}<0$ at the bottom, they are obstructed at the
interface; (c) When electrons diffuse within the MDA shown in (d),
where $B_{z}<0$ at the top and $B_{z}>0$ at the bottom, they pass
through the interface more easily.}
\end{figure}

\subsection{Electric transport and the Hall conductivity tensor \label{subsec:Electric-transport}}

Suppose the system is under a uniform electric field $\boldsymbol{E}=E_{x}\boldsymbol{e}_{x}+E_{y}\boldsymbol{e}_{y}$,
and the force field acting on $\alpha$-species is thus given by $\boldsymbol{F}_{\alpha}=q_{\alpha}\boldsymbol{E}$.
Assuming the system comprises two species of particles, electrons
and ions. The speed of the particles is constant for each species,
which is represented $u_{\alpha}\left(\boldsymbol{r}\right)=u_{\alpha}$.
Given that electrons move much faster than ions due to their lighter
mass, we assume the ions are stationary, i.e., $u_{i}=0$, where the
subscript $i$ refers to ions. The particle density is uniform throughout
the system, expressed as $n_{\alpha0}\left(\boldsymbol{r}\right)=n_{\alpha0}$.
The system satisfies the neutrality condition at equilibrium, which
is
\begin{equation}
\sum_{\alpha}q_{\alpha}n_{\alpha0}=-en_{e0}+q_{i}n_{i0}=0.
\end{equation}
In this scenario, the operator becomes $\hat{\boldsymbol{\mathcal{D}}}_{\alpha}=\left|\boldsymbol{v}\right|\partial/\partial\boldsymbol{r}+\left(q_{\alpha}\boldsymbol{E}/m_{\alpha}\right)\partial/\partial\left|\boldsymbol{v}\right|$,
and the perturbed distribution function $\delta f_{\alpha}^{\pm}$
is then given by

\begin{align}
\delta f_{e}^{\pm} & =\frac{en_{e0}}{2\pi m_{e}u_{e}}\frac{\partial\delta\left(\left|\boldsymbol{v}\right|-u_{e}\right)}{\partial\left|\boldsymbol{v}\right|}\boldsymbol{e}_{v}\cdot\boldsymbol{M}_{e}^{\pm}\cdot\boldsymbol{E},\label{eq:108}\\
\delta f_{i}^{\pm} & =0.\label{eq:109}
\end{align}
The electric current is calculated as

\begin{align}
\boldsymbol{J}^{\pm} & =\sum_{\alpha}q_{\alpha}\int f_{\alpha}\boldsymbol{v}d^{2}\boldsymbol{v}=\sum_{\alpha}q_{\alpha}\int\delta f_{\alpha}^{\pm}\boldsymbol{v}d^{2}\boldsymbol{v}=-e\int\delta f_{e}^{\pm}\boldsymbol{v}d^{2}\boldsymbol{v}\nonumber \\
 & =-\frac{e^{2}n_{e0}}{2\pi m_{e}u_{e}}\int\frac{\partial\delta\left(\left|\boldsymbol{v}\right|-u_{e}\right)}{\partial\left|\boldsymbol{v}\right|}\boldsymbol{v}\boldsymbol{e}_{v}\cdot\boldsymbol{M}_{e}^{\pm}\left(\boldsymbol{r},\left|\boldsymbol{v}\right|\right)\cdot\boldsymbol{E}d^{2}\boldsymbol{v}\nonumber \\
 & =-\frac{e^{2}n_{e0}}{2\pi m_{e}u_{e}}\int\frac{\partial\delta\left(\left|\boldsymbol{v}\right|-u_{e}\right)}{\partial\left|\boldsymbol{v}\right|}\left|\boldsymbol{v}\right|\boldsymbol{e}_{v}\boldsymbol{e}_{v}\cdot\boldsymbol{M}_{e}^{\pm}\left(\boldsymbol{r},\left|\boldsymbol{v}\right|\right)d^{2}\boldsymbol{v}\cdot\boldsymbol{E}\nonumber \\
 & =-\frac{e^{2}n_{e0}}{2\pi m_{e}u_{e}}\left[\int_{0}^{2\pi}\boldsymbol{e}_{v}\boldsymbol{e}_{v}d\xi\right]\cdot\left[\int_{0}^{+\infty}\frac{\partial\delta\left(\left|\boldsymbol{v}\right|-u_{e}\right)}{\partial\left|\boldsymbol{v}\right|}\boldsymbol{M}_{\alpha}^{\pm}\left(\boldsymbol{r},\left|\boldsymbol{v}\right|\right)\left|\boldsymbol{v}\right|^{2}d\left|\boldsymbol{v}\right|\right]\cdot\boldsymbol{E}\nonumber \\
 & =\frac{e^{2}n_{e0}}{2\pi m_{e}u_{e}}\left[\pi\boldsymbol{I}\right]\cdot\frac{\partial}{\partial\left|\boldsymbol{v}\right|}\mid_{\left|\boldsymbol{v}\right|=u_{e}}\left[\boldsymbol{M}_{\alpha}^{\pm}\left(\boldsymbol{r},\left|\boldsymbol{v}\right|\right)\left|\boldsymbol{v}\right|^{2}\right]\cdot\boldsymbol{E}\nonumber \\
 & =\frac{e^{2}n_{e0}}{2m_{e}u_{e}}\frac{\partial}{\partial u_{e}}\left[u_{e}^{2}\boldsymbol{M}_{\alpha}^{\pm}\left(\boldsymbol{r},u_{e}\right)\right]\cdot\boldsymbol{E},\label{eq:121-1}
\end{align}
where we used the following equation 
\begin{equation}
f\left(x\right)\delta'\left(x\right)=-f'\left(0\right)\delta\left(x\right)+f\left(0\right)\delta'\left(x\right).\label{eq:delta-property}
\end{equation}
Equation (\ref{eq:121-1}) can be rewritten as 
\begin{equation}
\boldsymbol{J}^{\pm}=\boldsymbol{\sigma}_{e}^{\pm}\cdot\boldsymbol{E},
\end{equation}
where 
\begin{equation}
\boldsymbol{\sigma}_{e}^{\pm}=\frac{e^{2}n_{e0}}{2m_{e}u_{e}}\frac{\partial}{\partial u_{e}}\left[u_{e}^{2}\boldsymbol{M}_{\alpha}^{\pm}\left(\boldsymbol{r},u_{e}\right)\right]\label{eq:conductivity tensor}
\end{equation}
is the Hall conductivity tensor of electrons. From Eq.\,(\ref{eq:conductivity tensor}),
we can see that 
\begin{equation}
\sigma_{eij}^{\pm}=-\sigma_{eji}^{\pm},\text{ for }i\ne j.\label{eq:125}
\end{equation}

\subsection{Heat transport \label{subsec:Heat-transport}}

As discussed in Subsec.\,\ref{subsec:Electric-transport}, we consider
only two species in our system: electrons and ions. The ions are stationary
($\boldsymbol{v}_{i}=0$, or in other words, they are cold ions).
The electron microscopic speed $u_{e}\left(\boldsymbol{r}\right)$
depends on the position $\boldsymbol{r}$, which implies that the
electron temperature and particle density is not uniform at equilibrium.
Assuming there is no equilibrium flow and no external force field,
we have $\boldsymbol{V}_{\alpha}=0$ and $\boldsymbol{F}_{\alpha}=0$.
Therefore, the operator becomes $\hat{\boldsymbol{\mathcal{D}}}=\left|\boldsymbol{v}\right|\partial/\partial\boldsymbol{r}$.
The perturbed distribution function $\delta f_{e}^{\pm}$ is then
calculated as
\begin{equation}
\delta f_{e}^{\pm}\left(\boldsymbol{r},\boldsymbol{v}\right)=-\frac{1}{2\pi}\boldsymbol{e}_{v}\cdot\boldsymbol{M}_{e}^{\pm}\left(\boldsymbol{r},\left|\boldsymbol{v}\right|\right)\cdot\frac{\partial}{\partial\boldsymbol{r}}\left[n_{e0}\left(\boldsymbol{r}\right)\delta\left(\left|\boldsymbol{v}\right|-u_{e}\left(\boldsymbol{r}\right)\right)\right].\label{eq:126}
\end{equation}
The electron stress tensor at equilibrium is calculated as follows
\begin{align}
\boldsymbol{p}_{e0} & \equiv m_{e}n_{e0}\left(\boldsymbol{r}\right)\left\langle \boldsymbol{w}_{e}\boldsymbol{w}_{e}\right\rangle _{e}=m_{e}n_{e}\left(\boldsymbol{r}\right)\left\langle \boldsymbol{v}\boldsymbol{v}\right\rangle _{e}\nonumber \\
 & =m_{e}n_{e0}\left(\boldsymbol{r}\right)\left\langle v_{x}^{2}\boldsymbol{e}_{x}\boldsymbol{e}_{x}+v_{x}v_{y}\left(\boldsymbol{e}_{x}\boldsymbol{e}_{y}+\boldsymbol{e}_{y}\boldsymbol{e}_{x}\right)+v_{y}^{2}\boldsymbol{e}_{y}\boldsymbol{e}_{y}\right\rangle _{e}\nonumber \\
 & =n_{e0}\left(\boldsymbol{r}\right)\left\langle \frac{1}{2}m_{e}\left|\boldsymbol{v}\right|^{2}\right\rangle _{e}\boldsymbol{I}=p_{e0}\left(\boldsymbol{r}\right)\boldsymbol{I},\label{eq:stress-tensor}
\end{align}
where 
\begin{equation}
p_{e0}\left(\boldsymbol{r}\right)=n_{e0}\left(\boldsymbol{r}\right)\left\langle \frac{1}{2}m_{e}\left|\boldsymbol{v}\right|^{2}\right\rangle _{e}=n_{e0}\left(\boldsymbol{r}\right)\left(\frac{1}{2}m_{e}u_{e}^{2}\left(\boldsymbol{r}\right)\right)\label{eq:pressure}
\end{equation}
is the electron pressure. In deriving Eq.\,(\ref{eq:stress-tensor}),
we used the relations
\begin{equation}
\left\langle v_{x}v_{y}\right\rangle _{\alpha}=0,\;\left\langle v_{x}^{2}\right\rangle _{\alpha}=\left\langle v_{y}^{2}\right\rangle _{\alpha}=\left\langle \frac{\left|\boldsymbol{v}\right|^{2}}{2}\right\rangle _{\alpha}.
\end{equation}
Using Eq.\,(\ref{eq:T/m}), the electron temperature is then given
by 
\begin{equation}
T_{e0}\left(\boldsymbol{r}\right)=\left\langle m_{e}\left|\boldsymbol{v}\right|^{2}/2\right\rangle _{e}=m_{e}u_{e}^{2}\left(\boldsymbol{r}\right)/2,\label{eq:electron-Temperture}
\end{equation}
and consequently, we have
\begin{equation}
\nabla u_{e}=\frac{1}{m_{e}u_{e}}\nabla T_{e0}.\label{eq:120}
\end{equation}
Substituting Eq.\,(\ref{eq:electron-Temperture}) into Eq.\,(\ref{eq:pressure}),
we obtain 
\begin{equation}
p_{e0}\left(\boldsymbol{r}\right)=n_{e0}\left(\boldsymbol{r}\right)T_{e0}\left(\boldsymbol{r}\right)\label{eq:132}
\end{equation}
which is the equation of state of ideal gas. Given that $\boldsymbol{V}_{\alpha0}=0$,
the scattering force $\boldsymbol{R}_{\alpha}\left(t,\boldsymbol{r}\right)$
is zero. Applying Eq.\,(\ref{eq:momentum-equation}), we have the
force balance equation 
\begin{equation}
\nabla p_{e0}=0.\label{eq:force-balance-1}
\end{equation}
Substituting Eq.\,(\ref{eq:132}) into Eq.\,$\left(\ref{eq:force-balance-1}\right),$we
obtain
\begin{equation}
\nabla n_{e0}=-\frac{n_{e0}}{T_{e0}}\nabla T_{e0}=-\frac{2n_{e0}}{m_{e}u_{e}^{2}}\nabla T_{e0},\label{eq:force-balance-2}
\end{equation}
where we used Eq.\,(\ref{eq:electron-Temperture}). By employing
Eq.\,(\ref{eq:heat-flux}), the electron heat flux $\boldsymbol{q}_{e}^{\pm}$
is determined as
\begin{align}
 & \boldsymbol{q}_{e}^{\pm}=\frac{1}{2}m_{e}n_{e0}\left\langle \boldsymbol{w}_{e}^{2}\boldsymbol{w}_{e}\right\rangle _{e}=\frac{1}{2}m_{\alpha}n_{\alpha}\left\langle \boldsymbol{v}^{2}\boldsymbol{v}\right\rangle _{\alpha}=\frac{1}{2}m_{e}\int\delta f_{e}^{\pm}\boldsymbol{v}^{2}\boldsymbol{v}d^{2}\boldsymbol{v}\nonumber \\
 & =-\frac{m_{e}}{4\pi}\int\left\{ \boldsymbol{e}_{v}\boldsymbol{e}_{v}\cdot\boldsymbol{M}_{e}^{\pm}\left(\boldsymbol{r},\left|\boldsymbol{v}\right|\right)\cdot\left[\left(\nabla n_{e0}\right)\delta\left(\left|\boldsymbol{v}\right|-u_{e}\right)-n_{e0}\left(\nabla u_{e}\right)\frac{\partial\delta\left(\left|\boldsymbol{v}\right|-u_{e}\right)}{\partial\left|\boldsymbol{v}\right|}\right]\right\} \left|\boldsymbol{v}\right|^{3}d^{2}\boldsymbol{v}\nonumber \\
 & =-\frac{m_{e}}{4\pi}\left[\int_{0}^{2\pi}\boldsymbol{e}_{v}\boldsymbol{e}_{v}d\xi\right]\cdot\left[\int_{0}^{+\infty}\left|\boldsymbol{v}\right|^{4}\boldsymbol{M}_{e}^{\pm}\left(\boldsymbol{r},\left|\boldsymbol{v}\right|\right)\delta\left(\left|\boldsymbol{v}\right|-u_{e}\right)d\left|\boldsymbol{v}\right|\right]\cdot\left(\nabla n_{e0}\right)\nonumber \\
 & +\frac{m_{e}}{4\pi}\left[\int_{0}^{2\pi}\boldsymbol{e}_{v}\boldsymbol{e}_{v}d\xi\right]\cdot\left[\int_{0}^{+\infty}\left|\boldsymbol{v}\right|^{4}\boldsymbol{M}_{e}^{\pm}\left(\boldsymbol{r},\left|\boldsymbol{v}\right|\right)\frac{\partial\delta\left(\left|\boldsymbol{v}\right|-u_{e}\right)}{\partial\left|\boldsymbol{v}\right|}d\left|\boldsymbol{v}\right|\right]\cdot n_{e0}\left(\nabla u_{e}\right)\nonumber \\
 & =-\frac{m_{e}u_{e}^{4}}{4}\boldsymbol{M}_{e}^{\pm}\left(\boldsymbol{r},u_{e}\right)\cdot\left(\nabla n_{e0}\right)-\frac{m_{e}n_{e0}}{4}\frac{\partial}{\partial u_{e}}\left[u_{e}^{4}\boldsymbol{M}_{e}^{\pm}\left(\boldsymbol{r},u_{e}\right)\right]\left(\nabla u_{e}\right).\label{eq:135}
\end{align}
The expression for the electron heat flux (\ref{eq:135}) can be reformulated
as
\begin{equation}
\boldsymbol{q}_{e}^{\pm}=-\boldsymbol{\kappa}_{e}^{\pm}\cdot\nabla T_{e0}
\end{equation}
by substituting Eqs.\,(\ref{eq:120}) and (\ref{eq:force-balance-2})
into Eq.\,(\ref{eq:135}). Here,
\begin{equation}
\boldsymbol{\kappa}_{e}^{\pm}=\frac{n_{e0}u_{e}}{4}\frac{\partial}{\partial u_{e}}\left[u_{e}^{2}\boldsymbol{M}_{e}^{\pm}\left(\boldsymbol{r},u_{e}\right)\right].\label{eq:thermal conductivity}
\end{equation}
represents the thermal Hall conductivity tensor. Similarly to Eq.\,(\ref{eq:125}),
we have 
\begin{equation}
\kappa_{eij}^{\pm}=\kappa_{eji}^{\pm},\text{ for }i\ne j.
\end{equation}
Since the macroscopic velocity $\boldsymbol{V}_{e}=0$, the energy
equation (\ref{eq:energy-equation}) simplifies to
\begin{equation}
\nabla\cdot\left(\boldsymbol{\kappa}_{e}^{\pm}\cdot\nabla T_{e0}\right)=0\label{eq:139}
\end{equation}
which corresponds to the stationary thermal transport equation.

By comparing Eqs.\,(\ref{eq:conductivity tensor}) and (\ref{eq:thermal conductivity}),
we find that
\begin{equation}
\boldsymbol{\kappa}_{e}^{\pm}=\frac{T_{e}}{e^{2}}\boldsymbol{\sigma}_{e}^{\pm},
\end{equation}
which corresponds to the Wiedemann--Franz law with the Lorenz number
$L=1$.

\section{Conclusions}

In this work, we explored the concept of metafields, a class of materials
that utilize local magnetic fields as the fundamental repeating units,
offering advantages over traditional metamaterials. Unlike conventional
metamaterials, where the repeating elements consist of fixed atomic
or molecular structures, metafields provide real-time tunability through
the external control of electric currents. This dynamic adaptability
opens new possibilities for controlling transport phenomena and manipulating
particle behavior.

We focused on a specific metafield system: the magnetic disk array
(MDA), which consists of magnetic disks (MDs) generating uniform magnetic
fields perpendicular to their surfaces. By analyzing the transport
properties of charged particles in the MDA, we explored various Hall
transport phenomena, including Hall diffusivity, Hall conductivity,
and thermal Hall effects. These Hall effects arise due to the antisymmetric
nature of the transport tensors, where longitudinal gradients can
induce transverse fluxes.

Using a combination of complex variable formulation, perturbation
methods, and Fourier analysis, we derived the key transport coefficients
such as the Hall diffusion tensor, Hall conductivity tensor, and thermal
Hall conductivity tensor. These analytical methods provided efficient
and straightforward solutions for understanding particle behavior
in the metafield system. Additionally, through simulations, we demonstrated
the occurrence of a one-way preferential diffusion phenomenon at the
interface of two MDAs with opposing field directions, revealing the
system\textquoteright s unique directional transport properties.

Overall, this study highlights the potential of metafields, particularly
MDAs, in advancing the design of tunable materials with precise control
over particle dynamics. The concept of metafields provides a novel
platform for further exploration in both theoretical and applied physics.
\begin{acknowledgments}
This work was supported by National Natural Science Foundation of
China (NSFC-12304138 and 12005141) and Natural Science Foundation
of Anhui Province of China (2308085QA12).
\end{acknowledgments}

\appendix

\section{proofs of some formulas \label{sec:proofs}}

The detail proofs of Eq.\,(\ref{eq:trajectory}) are given by follows.
We first separate the variables $v$ and $t$ of second equation of
Eq.\,(\ref{eq:Lorentz-force-complex}) to get 
\begin{equation}
\frac{dv}{v}=i\omega_{c}dt.\label{eq:dv/v}
\end{equation}
Equation (\ref{eq:dv/v}) can be directly integrated as
\begin{equation}
\mathrm{ln}v=i\omega_{c}t+\mathrm{const}
\end{equation}
or 
\begin{equation}
v\left(t\right)=v_{0}e^{i\omega_{c}t},\label{eq:v=00003Dv_0*exp(iw_ct)}
\end{equation}
where we have used the initial condition $v\left(0\right)=v_{0}$.
Using the first equation of Eq.\,(\ref{eq:Lorentz-force-complex})
and make another integration for Eq.\,(\ref{eq:v=00003Dv_0*exp(iw_ct)})
again, we can directly get Eq.\,(\ref{eq:trajectory}).

We now give the derivations of Eqs.\,(\ref{eq:e^alpha})-(\ref{eq:e^theta}).
The conjugation of Eq.\,(\ref{eq:basic formular}) is 
\begin{equation}
z^{*}=\frac{e^{-i\alpha}-1}{e^{-i\theta}-1}=\frac{\left(1-e^{i\alpha}\right)e^{i\theta}}{\left(1-e^{i\theta}\right)e^{i\alpha}}=z\frac{e^{i\theta}}{e^{i\alpha}},
\end{equation}
which can be also rewritten as 
\begin{equation}
z^{*}e^{i\alpha}=ze^{i\theta}.\label{eq:z*theta~zalpha}
\end{equation}
Equation (\ref{eq:z*theta~zalpha}) can be transformed into 
\begin{equation}
z^{*}e^{i\alpha}-z=z\left(e^{i\theta}-1\right)=e^{i\alpha}-1,\label{eq:proof-e^alpha}
\end{equation}
where Eq.\,(\ref{eq:basic formular}) is used in the last step. Equation
(\ref{eq:e^alpha}) can be directly derived using Eq.\,(\ref{eq:proof-e^alpha}).
Equation (\ref{eq:e^theta}) can be similarly derived by using 
\begin{equation}
\frac{e^{i\theta}-1}{e^{i\alpha}-1}=\frac{1}{z}.
\end{equation}

The proof of Eq.\,(\ref{eq:R dot V =00003D0}) is as follows
\begin{align}
0 & =\int\frac{1}{2}m_{\alpha}v^{2}\left(\frac{\partial f_{\alpha}}{\partial t}\right)_{c}d^{2}\boldsymbol{v}\nonumber \\
 & =\int\frac{1}{2}m_{\alpha}\left(w_{\alpha}+V_{\alpha}\right)^{2}\left(\frac{\partial f_{\alpha}}{\partial t}\right)_{c}d^{2}\boldsymbol{v}\nonumber \\
 & =\frac{1}{2}m_{\alpha}V_{\alpha}^{2}\int\left(\frac{\partial f_{\alpha}}{\partial t}\right)_{c}d^{2}\boldsymbol{v}+m_{\alpha}\boldsymbol{V}_{\alpha}\cdot\int\boldsymbol{w}\left(\frac{\partial f_{\alpha}}{\partial t}\right)_{c}d^{2}\boldsymbol{v}+\frac{1}{2}m_{\alpha}\int w_{\alpha}^{2}\left(\frac{\partial f_{\alpha}}{\partial t}\right)_{c}d^{2}\boldsymbol{v}\nonumber \\
 & =\boldsymbol{R}_{\alpha}\cdot\boldsymbol{V}_{\alpha}.\label{eq:proof-R dot V =00003D0}
\end{align}
Here, we used
\begin{equation}
Q_{\alpha}\equiv\frac{1}{2}m_{\alpha}\int w_{\alpha}^{2}\left(\frac{\partial f_{\alpha}}{\partial t}\right)_{c}d^{2}\boldsymbol{v}=0
\end{equation}
because the background MDA does not change the internal energy during
the collision process.

The proof of Eq.\,(\ref{eq:energy-equation}) proceeds as follows.
By combining Eqs.\,(\ref{eq:first-moment}) and (\ref{eq:collision-invariant-energy})
and substituting $\psi\left(\boldsymbol{v}\right)=m_{\alpha}\boldsymbol{v}^{2}/2$
into Eq.\,(\ref{eq:general-average-Eq}), we obtain
\begin{equation}
\frac{\partial}{\partial t}\left(\frac{1}{2}m_{\alpha}n_{\alpha}\boldsymbol{V}_{\alpha}^{2}+U_{\alpha}\right)+\nabla\cdot\left[\frac{1}{2}m_{\alpha}n_{\alpha}\left\langle \boldsymbol{v}^{2}\boldsymbol{v}\right\rangle _{\alpha}\right]-n_{\alpha}\left\langle \boldsymbol{F}_{\alpha}\cdot\boldsymbol{v}\right\rangle _{\alpha}=0.\label{eq:proof-energy-s1}
\end{equation}
The second term on the right-hand side of Eq.\,(\ref{eq:proof-energy-s1})
can be written as
\begin{equation}
n_{\alpha}\left\langle \boldsymbol{F}_{\alpha}\cdot\boldsymbol{v}\right\rangle _{\alpha}=n_{\alpha}\left\langle \boldsymbol{F}_{\alpha}^{\left(1\right)}\cdot\boldsymbol{v}\right\rangle _{\alpha}=n_{\alpha}\boldsymbol{F}_{\alpha}^{\left(1\right)}\cdot\left\langle \boldsymbol{v}\right\rangle _{\alpha}=n_{\alpha}\boldsymbol{F}_{\alpha}^{\left(1\right)}\cdot\boldsymbol{V}_{\alpha}.\label{eq:proof-energy-s2}
\end{equation}
The second term on the right-hand side of Eq.\,(\ref{eq:proof-energy-s1})
can be transformed into 
\begin{align}
 & \nabla\cdot\left[\frac{1}{2}m_{\alpha}n_{\alpha}\left\langle \boldsymbol{v}^{2}\boldsymbol{v}\right\rangle _{\alpha}\right]=\nabla\cdot\left\{ \frac{1}{2}m_{\alpha}n_{\alpha}\left[\left\langle \left(\boldsymbol{V}_{\alpha}+\boldsymbol{w}_{\alpha}\right)^{2}\left(\boldsymbol{V}_{\alpha}+\boldsymbol{w}_{\alpha}\right)\right\rangle _{\alpha}\right]\right\} \nonumber \\
 & =\nabla\cdot\left[\frac{1}{2}m_{\alpha}n_{\alpha}\left[\left\langle \left(\boldsymbol{V}_{\alpha}^{2}+2\boldsymbol{w}_{\alpha}\cdot\boldsymbol{V}_{\alpha}+\boldsymbol{w}_{\alpha}^{2}\right)\boldsymbol{V}_{\alpha}\right\rangle _{\alpha}+\left\langle \left(\boldsymbol{V}_{\alpha}^{2}+2\boldsymbol{w}_{\alpha}\cdot\boldsymbol{V}_{\alpha}+\boldsymbol{w}_{\alpha}^{2}\right)\boldsymbol{w}_{\alpha}\right\rangle _{\alpha}\right]\right]\nonumber \\
 & =\nabla\cdot\left[\frac{1}{2}m_{\alpha}n_{\alpha}\left[\boldsymbol{V}_{\alpha}^{2}\boldsymbol{V}_{\alpha}+\left\langle \boldsymbol{w}_{\alpha}^{2}\right\rangle _{\alpha}\boldsymbol{V}_{\alpha}+2\left\langle \boldsymbol{w}_{\alpha}\boldsymbol{w}_{\alpha}\right\rangle _{\alpha}\cdot\boldsymbol{V}_{\alpha}+\left\langle \boldsymbol{w}_{\alpha}^{2}\boldsymbol{w}_{\alpha}\right\rangle _{\alpha}\right]\right]\nonumber \\
 & =\nabla\cdot\left[\left(\frac{1}{2}m_{\alpha}n_{\alpha}\boldsymbol{V}_{\alpha}^{2}\right)\boldsymbol{V}_{\alpha}+U_{\alpha}\boldsymbol{V}_{\alpha}+\boldsymbol{p}_{\alpha}\cdot\boldsymbol{V}_{\alpha}+\boldsymbol{q}_{\alpha}\right],\label{eq:proof-energy-s3}
\end{align}
where we used the relation
\begin{equation}
\left\langle \boldsymbol{w}_{\alpha}\cdot\boldsymbol{\Phi}\left(\boldsymbol{V}_{\alpha}\right)\right\rangle _{\alpha}=\boldsymbol{\Phi}\left(\boldsymbol{V}_{\alpha}\right)\cdot\left\langle \boldsymbol{w}_{\alpha}\right\rangle _{\alpha}=0,
\end{equation}
and where $\boldsymbol{\Phi}\left(\boldsymbol{V}_{\alpha}\right)$
is an arbitrary function of $\boldsymbol{V}_{\alpha}$. Substituting
Eqs.\,(\ref{eq:proof-energy-s2}) and (\ref{eq:proof-energy-s3})
into Eq.\,(\ref{eq:proof-energy-s1}), we have 
\begin{align}
 & \frac{\partial}{\partial t}\left(\frac{1}{2}m_{\alpha}n_{\alpha}\boldsymbol{V}_{\alpha}^{2}\right)+\nabla\cdot\left[\left(\frac{1}{2}m_{\alpha}n_{\alpha}\boldsymbol{V}_{\alpha}^{2}\right)\boldsymbol{V}_{\alpha}\right]+\nabla\cdot\left(\boldsymbol{p}_{\alpha}\cdot\boldsymbol{V}_{\alpha}\right)\nonumber \\
 & +\frac{\partial U_{\alpha}}{\partial t}+\nabla\cdot\left(U_{\alpha}\boldsymbol{V}_{\alpha}\right)+\nabla\cdot\boldsymbol{q}_{\alpha}=n_{\alpha}\boldsymbol{F}_{\alpha}^{\left(1\right)}\cdot\boldsymbol{V}_{\alpha}.\label{eq:proof-energy-s4}
\end{align}
The first three terms on the left-hand side of Eq.\,(\ref{eq:proof-energy-s4})
can be transformed into 
\begin{align}
 & \frac{\partial}{\partial t}\left(\frac{1}{2}m_{\alpha}n_{\alpha}\boldsymbol{V}_{\alpha}^{2}\right)+\nabla\cdot\left[\left(\frac{1}{2}m_{\alpha}n_{\alpha}\boldsymbol{V}_{\alpha}^{2}\right)\boldsymbol{V}_{\alpha}\right]+\nabla\cdot\left(\boldsymbol{p}_{\alpha}\cdot\boldsymbol{V}_{\alpha}\right)\nonumber \\
 & =\left(\frac{1}{2}m_{\alpha}\boldsymbol{V}_{\alpha}^{2}\frac{\partial n_{\alpha}}{\partial t}+m_{\alpha}n_{\alpha}\boldsymbol{V}_{\alpha}\cdot\frac{\partial\boldsymbol{V}_{\alpha}}{\partial t}\right)+n_{\alpha}\boldsymbol{V}_{\alpha}\cdot\nabla\left(\frac{1}{2}m_{\alpha}\boldsymbol{V}_{\alpha}^{2}\right)\nonumber \\
 & +\left(\frac{1}{2}m_{\alpha}\boldsymbol{V}_{\alpha}^{2}\right)\nabla\cdot\left(n_{\alpha}\boldsymbol{V}_{\alpha}\right)+\left(\nabla\cdot\boldsymbol{p}_{\alpha}\right)\cdot\boldsymbol{V}_{\alpha}+\boldsymbol{p}_{\alpha}:\nabla\boldsymbol{V}_{\alpha}\nonumber \\
 & =\left(m_{\alpha}n_{\alpha}\frac{\partial\boldsymbol{V}_{\alpha}}{\partial t}+m_{\alpha}n_{\alpha}\boldsymbol{V}_{\alpha}\cdot\nabla\boldsymbol{V}_{\alpha}\right)\cdot\boldsymbol{V}_{\alpha}+\left(\frac{1}{2}m_{\alpha}\boldsymbol{V}_{\alpha}^{2}\right)\nabla\cdot\left(n_{\alpha}\boldsymbol{V}_{\alpha}\right)\nonumber \\
 & +\left(\nabla\cdot\boldsymbol{p}_{\alpha}\right)\cdot\boldsymbol{V}_{\alpha}+\boldsymbol{p}_{\alpha}:\nabla\boldsymbol{V}_{\alpha}\nonumber \\
 & =\left[\frac{\partial}{\partial t}\left(m_{\alpha}n_{\alpha}\boldsymbol{V}_{\alpha}\right)+\nabla\cdot\left(m_{\alpha}n_{\alpha}\boldsymbol{V}_{\alpha}\boldsymbol{V}_{\alpha}\right)+\left(\nabla\cdot\boldsymbol{p}_{\alpha}\right)-m_{\alpha}\boldsymbol{V}_{\alpha}\left(\frac{\partial n_{\alpha}}{\partial t}+\nabla\cdot\left(n_{\alpha}\boldsymbol{V}_{\alpha}\right)\right)\right]\cdot\boldsymbol{V}_{\alpha}\nonumber \\
 & +\boldsymbol{p}_{\alpha}:\nabla\boldsymbol{V}_{\alpha}\nonumber \\
 & =\left(\boldsymbol{R}_{\alpha}+n_{\alpha}\boldsymbol{F}_{\alpha}\right)\cdot\boldsymbol{V}_{\alpha}+\boldsymbol{p}_{\alpha}:\nabla\boldsymbol{V}_{\alpha}=n_{\alpha}\boldsymbol{F}_{\alpha}^{\left(1\right)}\cdot\boldsymbol{V}_{\alpha}+\boldsymbol{p}_{\alpha}:\nabla\boldsymbol{V}_{\alpha},\label{eq:proof-energy-s5}
\end{align}
where we used the Eqs.\,(\ref{eq:Force1+2}), (\ref{eq:R dot V =00003D0}),
(\ref{eq:continuity-equation}) and (\ref{eq:momentum-equation})
in the last three steps. Substituting Eq.\,(\ref{eq:proof-energy-s5})
into Eq.\,(\ref{eq:proof-energy-s4}) we finally arrive at the energy
equation (\ref{eq:energy-equation}).

\section{Properties of $K_{m}^{\pm}\left(x\right)$ function \label{sec:Km(x) function}}

Using Eq.\,(\ref{eq:Km^(pm)}), $K_{m}^{\pm}\left(x\right)$ can
be expressed as 
\begin{align}
K_{m}^{+}\left(x\right) & =\int_{0}^{\pi}\left[1-\left(\frac{x-e^{i\varphi}}{x-e^{-i\varphi}}\right)^{m}\right]\mathrm{sin\varphi}d\varphi,\label{eq:Km(x)-+}\\
K_{m}^{-}\left(x\right) & =\int_{0}^{\pi}\left[1-\left(\frac{x+e^{i\varphi}}{x+e^{-i\varphi}}\right)^{m}\right]\mathrm{sin\varphi}d\varphi.\label{eq:Km(x)--}
\end{align}
Based on Eqs.\,(\ref{eq:Km(x)-+}) and (\ref{eq:Km(x)--}), we have
the following propositions regarding the function $K_{m}^{\pm}\left(x\right)$.
\begin{prop}
$\forall m\in\mathbb{Z}$ and $x\in\mathbb{R},$ we have
\begin{equation}
K_{m}^{\pm}\left(x\right)=\left[K_{-m}^{\pm}\left(x\right)\right]^{*}.
\end{equation}
\end{prop}
\begin{proof}
Using to Eqs.\,(\ref{eq:Km(x)-+}) and (\ref{eq:Km(x)--}), we can
directly calculate $\left[K_{-m}^{\pm}\left(x\right)\right]^{*}$
as follows
\begin{align}
 & \left[K_{-m}^{\pm}\left(x\right)\right]^{*}=\int_{0}^{\pi}\left[1-\left(\frac{x-\left(\pm\right)e^{i\varphi}}{x-\left(\pm\right)e^{-i\varphi}}\right)^{-m}\right]^{*}\mathrm{sin\varphi}d\varphi=\int_{0}^{\pi}\left[1-\left(\frac{x-\left(\pm\right)e^{-i\varphi}}{x-\left(\pm\right)e^{i\varphi}}\right)^{m}\right]^{*}\mathrm{sin\varphi}d\varphi\nonumber \\
 & =\int_{0}^{\pi}\left[1-\left(\frac{x-\left(\pm\right)e^{i\varphi}}{x-\left(\pm\right)e^{-i\varphi}}\right)^{m}\right]\mathrm{sin\varphi}d\varphi=K_{m}^{\pm}\left(x\right).
\end{align}
\end{proof}
\begin{prop}
$\forall m\in\mathbb{Z}$ and $x\in\mathbb{R},$ we have
\begin{equation}
\left[K_{m}^{+}\left(x\right)\right]^{*}=K_{m}^{-}\left(x\right).\label{eq:148}
\end{equation}
\end{prop}
\begin{proof}
By introducing the new variable $\varphi'=-\varphi+\pi$, we can rewrite
Eq.\,(\ref{eq:Km(x)-+}) as
\begin{align}
 & K_{m}^{+}\left(x\right)=\int_{0}^{\pi}\left[1-\left(\frac{x-e^{i\left(-\varphi'+\pi\right)}}{x-e^{-i\left(-\varphi'+\pi\right)}}\right)^{m}\right]\mathrm{sin\left(-\varphi'+\pi\right)}d\left(-\varphi'+\pi\right)\nonumber \\
 & =-\int_{\pi}^{0}\left[1-\left(\frac{x+e^{-i\varphi'}}{x+e^{i\varphi'}}\right)^{m}\right]\mathrm{sin\varphi'}d\varphi'=\int_{0}^{\pi}\left[1-\left(\frac{x+e^{-i\varphi}}{x+e^{i\varphi}}\right)^{m}\right]\mathrm{sin\varphi}d\varphi.\label{eq:149}
\end{align}
Finally, using Eq.\,(\ref{eq:149}), we can easily prove Eq.\,(\ref{eq:148})
by directly calculating the conjugate.
\end{proof}
\begin{prop}
$\forall m\in\mathbb{Z}$ and $x\in\mathbb{R},$ we have

\begin{equation}
K_{m}^{\pm}\left(x\right)=\left[K_{m}^{\pm}\left(-x\right)\right]^{*}.\label{eq:151}
\end{equation}
\end{prop}
\begin{proof}
We first rewrite $\left[K_{m}^{\pm}\left(-x\right)\right]^{*}$ as
\begin{equation}
\left[K_{m}^{\pm}\left(-x\right)\right]^{*}=\int_{0}^{\pi}\left[1-\left(\frac{x+\left(\pm\right)e^{-i\varphi}}{x+\left(\pm\right)e^{i\varphi}}\right)^{m}\right]\mathrm{sin\varphi}d\varphi\label{eq:163}
\end{equation}
by using Eqs.\,(\ref{eq:Km(x)-+}) and (\ref{eq:Km(x)--}). By introducing
the new variable $\varphi'=-\varphi+\pi$, equation (\ref{eq:163})
can be transformed as 
\begin{align}
 & \left[K_{m}^{\pm}\left(-x\right)\right]^{*}=\int_{0}^{\pi}\left[1-\left(\frac{x+\left(\pm\right)e^{-i\left(-\varphi'+\pi\right)}}{x+\left(\pm\right)e^{i\left(-\varphi'+\pi\right)}}\right)^{m}\right]\mathrm{sin\left(-\varphi'+\pi\right)}d\left(-\varphi'+\pi\right)\nonumber \\
 & =-\int_{\pi}^{0}\left[1-\left(\frac{x-\left(\pm\right)e^{i\varphi'}}{x-\left(\pm\right)e^{-i\varphi'}}\right)^{m}\right]\mathrm{sin\varphi'}d\varphi'=\int_{0}^{\pi}\left[1-\left(\frac{x-\left(\pm\right)e^{i\varphi}}{x-\left(\pm\right)e^{-i\varphi}}\right)^{m}\right]\mathrm{sin\varphi}d\varphi\nonumber \\
 & =K_{m}^{\pm}\left(x\right).\label{eq:164}
\end{align}
\end{proof}
\begin{figure}
\includegraphics[scale=0.7]{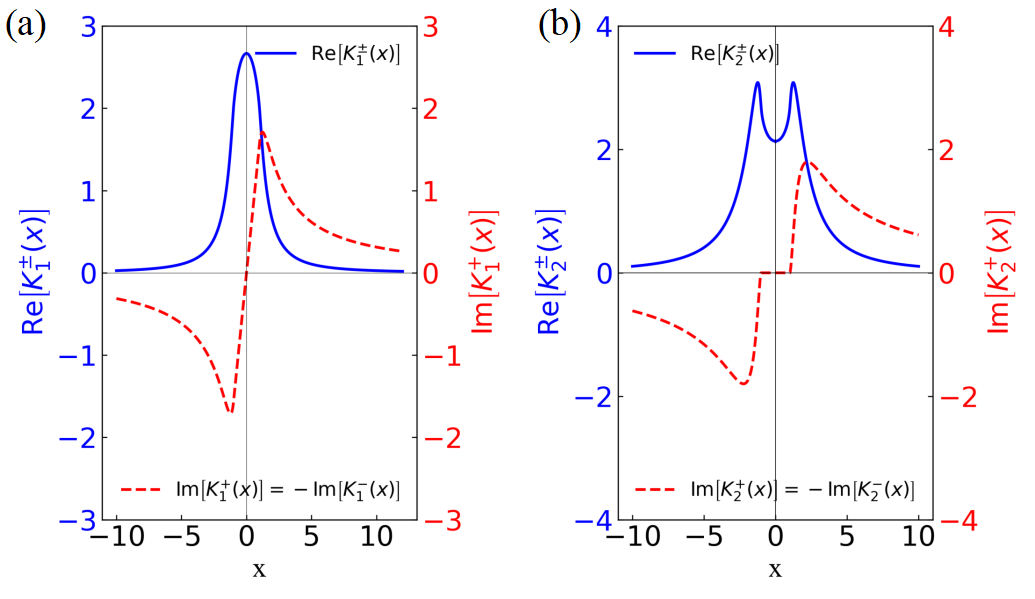}

\caption{\label{fig:K1,2(X)}Real and imaginary parts of the function $K_{1}^{\pm}\left(x\right)$
are shown in (a), and those of $K_{2}^{\pm}\left(x\right)$ are shown
in (b).}
\end{figure}

For $m=1$ and $2$, the function $K_{1}^{\pm}\left(x\right)$ and
$K_{2}^{\pm}\left(x\right)$ can be directly calculated in explicit
form as
\begin{align}
 & \mathrm{Re}\left[K_{1}^{+}\left(x\right)\right]=\mathrm{Re}\left[K_{1}^{-}\left(x\right)\right]=\frac{1}{2x^{3}}\left[2\left(x+x^{3}\right)-\left(x^{2}-1\right)^{2}\mathrm{ln}\left|\frac{x+1}{x-1}\right|\right]\\
 & \mathrm{Im}\left[K_{1}^{+}\left(x\right)\right]=-\mathrm{Im}\left[K_{1}^{-}\left(x\right)\right]=\begin{cases}
\frac{\pi}{2}x, & -1\leqslant x\leqslant1,\\
\frac{\pi}{2}\left(\frac{2}{x}-\frac{1}{x^{3}}\right), & \left|x\right|>1,
\end{cases}\\
 & \mathrm{Re}\left[K_{2}^{+}\left(x\right)\right]=\mathrm{Re}\left[K_{2}^{-}\left(x\right)\right]=\frac{1}{3x^{5}}\left[4x\left(5x^{2}-3\right)+6\left(x^{2}-1\right)^{2}\mathrm{ln}\left|\frac{x+1}{x-1}\right|\right],\\
 & \mathrm{Im}\left[K_{2}^{+}\left(x\right)\right]=-\mathrm{Im}\left[K_{2}^{-}\left(x\right)\right]=\begin{cases}
0, & -1\leqslant x\leqslant1,\\
\frac{2\pi}{x^{5}}\left(x^{2}-1\right)^{2}, & \left|x\right|>1.
\end{cases}
\end{align}
To compare the functions $K_{1}^{\pm}\left(x\right)$ and $K_{2}^{\pm}\left(x\right)$,
we plot their real and imaginary parts, as shown in Fig.\,\ref{fig:K1,2(X)}.
The real parts of both functions are even, while their imaginary parts
are odd, which corresponds to the properties described in Eq.\,(\ref{eq:151}).

\bibliographystyle{apsrev4-1}
\bibliography{Hall_MDA}

\end{document}